\documentclass[english]{IEEEtran}

\ifCLASSOPTIONcompsoc
  \usepackage[caption=false,font=footnotesize,labelfont=sf,textfont=sf]{subfig}
\else
  \usepackage[caption=false,font=footnotesize]{subfig}
\fi
\usepackage{url}
\ifCLASSOPTIONcompsoc
  \usepackage[nocompress]{cite}
\else
  \usepackage{cite}
\fi

\usepackage[T1]{fontenc}
\usepackage{float}
\usepackage{amsmath}
\usepackage{amssymb}
\usepackage{graphicx}
\usepackage{setspace}

\makeatletter

\floatstyle{ruled}
\newfloat{algorithm}{tbp}{loa}
\providecommand{\algorithmname}{Algorithm}
\floatname{algorithm}{\protect\algorithmname}

\usepackage{amsthm}

\newtheorem{proposition}{Proposition}
\newtheorem{corollary}{Corollary}

\usepackage{algorithmic}

\usepackage{babel}
\usepackage{upquote}

\usepackage{babel}

\usepackage{booktabs}\usepackage{bm}\usepackage{cite}

\makeatother

\usepackage{babel}
\begin{document}

\title{Online Placement of Multi-Component Applications in Edge Computing Environments}

\author{\IEEEauthorblockN{Shiqiang Wang, \IEEEmembership{Member,~IEEE,} Murtaza
Zafer, \IEEEmembership{Member,~IEEE,} and Kin K. Leung, \IEEEmembership{Fellow,~IEEE}}
\thanks{
This research was sponsored in part by the U.S. Army Research Laboratory and the U.K. Ministry of Defence and was accomplished under Agreement Number W911NF-06-3-0001 and W911NF-16-3-0001. The views and conclusions contained in this document are those of the author(s) and should not be interpreted as representing the official policies, either expressed or implied, of the U.S. Army Research Laboratory, the U.S. Government, the U.K. Ministry of Defence or the U.K. Government. The U.S. and U.K. Governments are authorized to reproduce and distribute reprints for Government purposes notwithstanding any copyright notation hereon.

S. Wang is with IBM T. J. Watson Research Center, Yorktown Heights, NY, United States. Email: wangshiq@us.ibm.com

M. Zafer is with Nyansa Inc., Palo Alto, CA, United States. E-mail: murtaza.zafer.us@ieee.org

K. K. Leung is with the Department of Electrical and Electronic Engineering, Imperial College London, United Kingdom. E-mail: kin.leung@imperial.ac.uk

Part of the material presented in this paper appeared in S. Wang's Ph.D. thesis \cite{wang2015dynamic}.

This is the author's version of the paper accepted for publication in IEEE Access,  DOI: 10.1109/ACCESS.2017.2665971.  \newline
\textcopyright 2017 IEEE. Personal use of this material is permitted. Permission from IEEE must be obtained for all other uses, in any current or future media, including reprinting/republishing this material for advertising or promotional purposes, creating new collective works, for resale or redistribution to servers or lists, or reuse of any copyrighted component of this work in other works.
}
}

\maketitle

\begin{abstract}
Mobile edge computing is a new cloud computing paradigm which makes use of small-sized edge-clouds to provide real-time services to users. These mobile edge-clouds (MECs) are located in close proximity to users, thus enabling users to seamlessly access applications running on MECs. Due to the co-existence of the core (centralized) cloud, users, and one or multiple layers of MECs, an important problem is to decide where (on which computational entity) to place different components of an application. This problem, known as the application or workload placement problem, is notoriously hard, and therefore, heuristic algorithms without performance guarantees are generally employed in common practice, which may unknowingly suffer from poor performance as compared to the optimal solution. In this paper, we address the application placement problem and focus on developing algorithms with provable performance bounds. We model the user application as an application graph and the physical computing system as a physical graph, with resource demands/availabilities annotated on these graphs. We first consider the placement of a linear application graph and propose an algorithm for finding its optimal solution. Using this result, we then generalize the formulation and obtain online approximation algorithms with polynomial-logarithmic (poly-log) competitive ratio for tree application graph placement. We jointly consider node and link assignment, and incorporate multiple types of computational resources at nodes.
\end{abstract}


\section{Introduction}
\label{sec:intro}

Mobile applications relying on cloud computing became increasingly popular in the recent years \cite{ViewOfCloud,bahl2012advancing}. 
Different from traditional standalone applications that run solely on a mobile device, a cloud-based application has one or multiple components running in the cloud, which are connected to another component running on the handheld device and they jointly constitute an application accessible to the mobile user. Examples of cloud-based mobile applications include map, storage, and video streaming services \cite{abe2013vtube,ha2013WearableCognitiveAssistance}. They all require high data processing/storage capability that cannot be satisfied on handheld devices alone, thus it is necessary to run part of the application in the cloud.

Traditionally, clouds are located in centralized data-centers. One problem with cloud-based applications is therefore the long-distance communication between the user device and the cloud, which may cause intermittent connectivity and long latency that cannot satisfy the requirements of emerging interactive applications such as real-time face recognition and online gaming \cite{CommMagEdgeComput}. To tackle this issue, \emph{mobile edge-cloud (MEC)} has been proposed recently \cite{IBMWhitepaper,ETSIWhitepaper}. The idea is to have small cloud-like entities (i.e., MECs) deployed at the edge of communication networks, which can run part or all of the application components. These MECs are located close to user locations, enabling users to have seamless and low-latency access to cloud services. For example, they can co-locate with edge devices such as Wi-Fi access points or cellular base stations (BSs), as shown in Fig. \ref{chap:intro:fig:scenario}, forming up a hierarchy together with the centralized cloud and mobile users. The concept of MEC is similar to cloudlet \cite{CloudletHostile}, fog computing \cite{bonomi2012fog,peng2015fog}, follow me cloud \cite{FollowMeGC2013}, and small cell cloud \cite{becvar2014pimrc}.

\begin{figure}
\center{\includegraphics[width=1\linewidth]{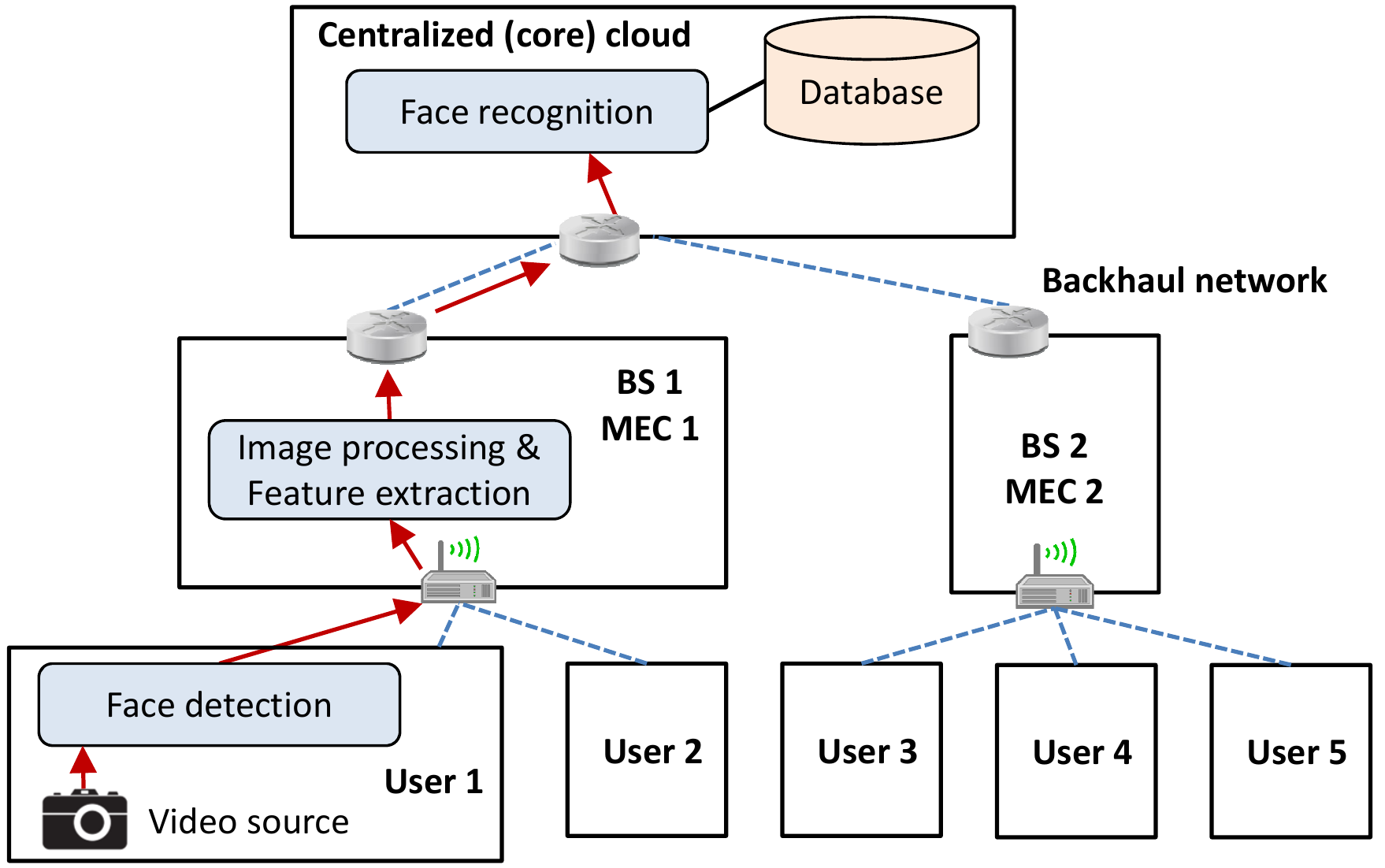}}

\protect\caption{Application scenario with mobile edge-clouds (MECs). Example scenario with face recognition application, where the dashed lines stand for physical communication links and red arrows stand for the data transmission path.}
\label{chap:intro:fig:scenario}
\end{figure}

Although MECs are promising, there are limitations. In particular, they have a significantly lower processing and storage capability compared to the core (centralized) cloud, thus it is usually infeasible to completely abandon the core cloud and run everything on MECs. An important problem is therefore to decide where (i.e., whether on the core cloud, MEC, or mobile device) to place different processing and storage components of an application. This is referred to as the \emph{application placement problem}, which is a non-trivial problem as illustrated by the example below.

\subsection{Motivating Example}
\label{subsec:MotivatingExample}
Consider an application which recognizes faces from a real-time video stream captured by the camera of a hand-held device. As shown in Fig.~\ref{chap:intro:fig:scenario}, we can decompose this application into one storage component (the database) and three different processing components including face detection (FD), image processing and feature extraction (IPFE), and face recognition (FR). The FD component finds areas of an image (a frame of the video stream) that contains faces. This part of image is sent to IPFE for further processing. The main job of IPFE is to filter out noise in the image and extract useful features for recognizing the person from its face. These features are sent to FR for matching with a large set of known features of different persons' faces stored in the database.

Fig.~\ref{chap:intro:fig:scenario} shows one possible placement of FD, IPFE, FR, and the database onto the hierarchical cloud architecture. This can be a good placement in some cases, but may not be a good placement in other cases. 

For example, the benefit of running FD on the mobile device instead of MEC is that it reduces the amount of data that need to be transferred between the mobile device and MEC. However, in cases where the mobile device's processing capability is strictly limited but there is a reasonably high bandwidth connection between the mobile device and MEC, it is can be good to place FD on the MEC.
Having the database in the core cloud can be beneficial because it can contain a large amount of data infeasible for the MEC to store. In this case, FR should also be in the core cloud because it needs to frequently query the database. 
However, if the database is relatively small and has locally generated contents, we may want to place the database and FR onto the MEC instead of the core cloud, as this reduces the backhaul network load.

We see that even with this simple application, it is non-straightforward to conceptually find the best placement, while many realistic applications such as streaming, multicasting, and data aggregation \cite{ChuMulticast,DataAggregationImpactWSN,ConcealedDataAggregation} are much more complex. 
We also note that MECs can be attached to devices at different cellular network layers \cite{ETSIWhitepaper}, yielding a hierarchical cloud structure with more than three layers.
Meanwhile, there usually exist multiple applications that are instantiated at the cloud system over time. All these aspects motivate us to consider the application placement problem in a rigorous optimization framework where applications arrive in an \emph{online} manner, i.e., they arrive sequentially over time and we do not have knowledge on the characteristics of future applications at any point of time.

We will abstract the application placement problem as the problem of placing application graphs, which represent application components and the communication among these components, onto a physical graph, which represents the computing devices and communication links in the physical system, as shown in Fig. \ref{fig:intro}. 
For example, the face recognition application above can be abstracted as an application graph with four nodes connected in a chain (line), where each of the nodes represents the database, FR, IPFE, and FD in sequential order.
The detailed problem formulation will be presented in Section \ref{sec:ProblemFormulation}.

\begin{figure}
\center{\includegraphics[width=1\columnwidth]{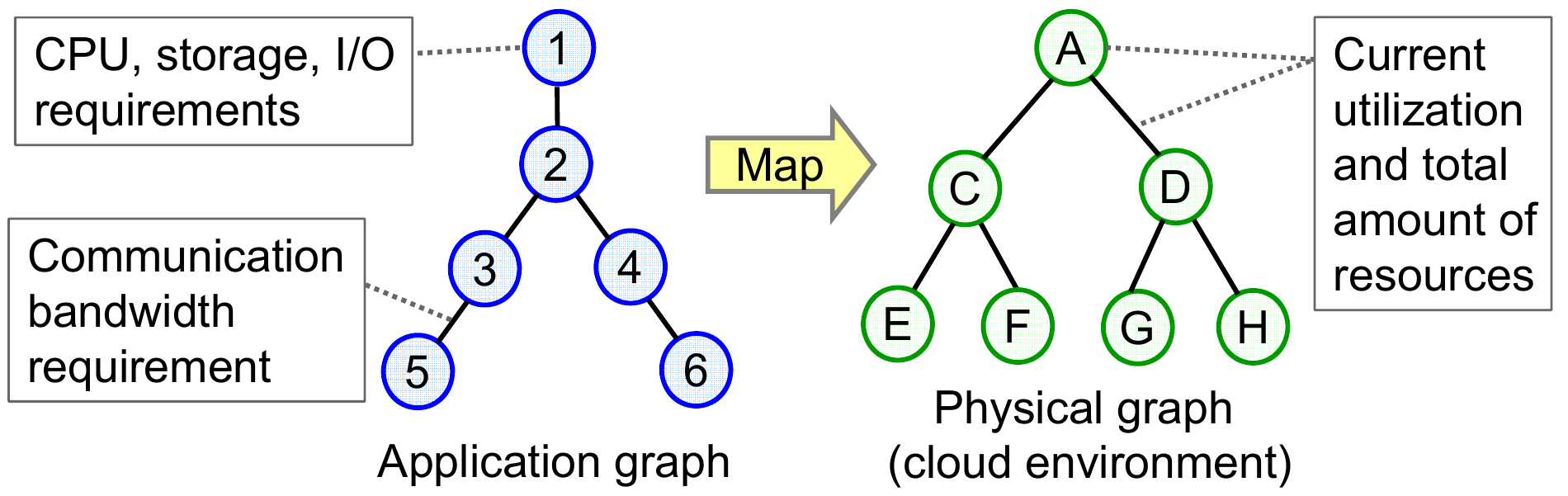}}
\caption{The application placement problem.}
\label{fig:intro} 
\end{figure}

\subsection{Related Work}
\label{sub:relatedWork}
A body of existing work on application placement and scheduling in MECs has considered applications with two components, one running on a cloud (can be either MEC or core cloud) and the other running on the user \cite{FollowMeGC2013, MDPFollowMeICC2014, wang2015IFIPNetworking, urgaonkar2015performance}.
Another body of existing work, which is also known under the term ``cloud offloading'', usually involves only two physical computing entities (i.e., the mobile device and the cloud)  \cite{cloudOffload1,cloudOffload2,ISIT2016}. 
Multi-component applications that can be deployed accross one or multiple levels of MECs and core cloud(s) have not been considered, whereas such applications widely exist in practice because MEC servers can locate at multiple network equipments in different hierarchical levels~\cite{ETSIWhitepaper}.

The multi-component application placement problem has been studied mainly in data-center settings.
Because this problem is NP-hard even for simple graphs (as we discuss later), a common practice is to employ heuristic algorithms without performance guarantees \cite{VNEmbeddingSurvey, refEfficientPlacementVI}, which may unknowingly suffer from poor performance as compared to the optimal solution. Only a very limited amount of existing work followed a rigorous theoretical framework from approximation algorithms \cite{bookApproxAlg} and competitive analysis \cite{bookOnlineComputation}, and proposed approximation algorithms (i.e., approximately optimal algorithms)
with provable approximation/competitive ratios\footnote{For a minimization problem, the \emph{competitive ratio} is defined as an upper bound of the online approximation algorithm's cost to the true optimal cost that can be obtained from an offline placement, where the offline placement considers all application graphs simultaneously instead of considering them arriving over time. The definition of approximation ratio is the same but it is for offline problems.} for the application placement problem, in particular when it involves both node and link placements.

In \cite{refViNEYard}, the authors
proposed an algorithm for minimizing the sum cost while considering load balancing, which has an approximate
approximation ratio of $O(N)$, where $N$ is the number of nodes
in the physical graph. The algorithm is based on linear program (LP) relaxation, and only allows one node in
each application graph to be placed on a particular physical node; thus, excluding
server resource sharing among different nodes in one application graph. 
It is shown that the approximation ratio of this algorithm is $O(N)$, which is trivial because one would achieve the same approximation ratio when placing the whole application graph onto a single physical node instead of distributing it across the whole physical graph.


A theoretical work in \cite{PODC2011} proposed
an algorithm with $N{}^{O(D)}$ time-complexity and an approximation
ratio of $\delta=O(D^{2}\log(ND))$ for placing a tree application graph with $D$ levels of nodes onto a physical
graph. It uses LP relaxation and its goal is to minimize the sum cost. Based on this algorithm, the authors presented an online
algorithm for minimizing the maximum load on each node and link, which
is $O(\delta\log(N))$-competitive when the application lifetimes are equal. 
The LP formulation in \cite{PODC2011} is complex and requires $N{}^{O(D)}$ variables and constraints. This means when $D$ is not a constant, the space-complexity (specifying the required memory size of the algorithm) is exponential in $D$.

Another related theoretical work which proposed an LP-based
method for offline placement of paths into trees in data-center networks was reported in \cite{PathsIntoTrees}.
Here, the application nodes can only be placed onto the leaves of
a tree physical graph, and the goal is to minimize link congestion. In our problem, the application nodes are distributed across users, MECs, and core cloud, thus they should not be only placed at the leaves of a tree so the problem formulation in \cite{PathsIntoTrees} is inapplicable to our scenario. Additionally, \cite{PathsIntoTrees} only focuses on minimizing link congestion. The load balancing of nodes is not considered as part of the objective; only the capacity limits of nodes are considered.


Some other related work focuses on graph partitioning, such as \cite{AllocDistributedCloud} and \cite{OptApproxLatencyMin}, where the physical graph is defined as a complete graph with edge costs associated with the distance or latency between physical servers. Such an abstraction combines multiple network links into one (abstract) physical edge, which may hide the actual status of individual links along a path.  

A related problem that has emerged recently is the service chain embedding problem \cite{serviceChainInfocom2016,rost2016service,Lukovszki2015}. Motivated by network function virtualization (NFV) applications, the goal is to place a linear application graph between fixed source and destination physical nodes, so that a series of operations are performed on data packets sent from the source to the destination. Within this body of work, only \cite{Lukovszki2015} has studied the competitive ratio of online placement, which, however, does not consider link placement optimization.

One important aspect to note is that most existing work, including \cite{refViNEYard,PathsIntoTrees,AllocDistributedCloud,OptApproxLatencyMin,serviceChainInfocom2016,rost2016service},  do not specifically consider the online operation of the algorithms. Although some of them implicitly claim that one can apply the algorithm repeatedly for each newly arrived application, the competitive ratio of such procedure is unclear. 
To the best of our knowledge, \cite{PODC2011} is the only work that studied the competitive ratio of the online application placement problem that considers \emph{both node and link placements}.

\subsection{Our Approach}
\label{sub:contributionsInIntro}
In this paper, we focus on the MEC context and propose algorithms for solving the online application placement problem with provable competitive ratios. Different from \cite{PODC2011}, our approach is \emph{not} based on LP relaxation. 
Instead, our algorithms are built upon a baseline algorithm that provides an \emph{optimal} solution to the placement of a linear application graph (i.e., an application graph that is a line). This is an important novelty in contrast to \cite{PODC2011} where no optimal solution was presented for any scenario. Many applications expected to run in an MEC environment can be abstracted as hierarchical graphs, and the simplest case of such a hierarchical graph is a line, such as the face recognition example in Section \ref{subsec:MotivatingExample}. Therefore, the placement of a linear application graph is an important problem in the context of MECs.

Another novelty in our work, compared to \cite{PODC2011} and most other approaches based on LP relaxation, is that our solution approach is \emph{decomposable} into multiple small building blocks. This makes it easy to extend our proposed algorithms to a distributed solution in the future, which would be very beneficial for reducing the amount of necessary control information exchange among different cloud entities in a distributed cloud environment containing MECs. This decomposable feature also makes it easier to use these algorithms as a sub-procedure for solving a larger problem. 

It is also worth noting that the analytical methodology we use in this paper is new compared to existing techniques such as LP relaxation, thus we enhance the set of tools for online optimization.
The theoretical analysis in this paper also provides insights on the features and difficulty of the problem, which can guide future practical implementations. In addition, the proposed algorithms themselves are relatively easy to implement in practice.



\subsection{Main Results}

\label{sec:mainResults}

We propose non-LP based approximation algorithms for
online application placement in this paper. The general problem of application
placement is hard to approximate \cite{PODC2011,PathsIntoTrees,serviceChainInfocom2016,QAPApprox}. For example, \cite{serviceChainInfocom2016} has shown that theoretically, there exists no polynomial-time approximation algorithm with bounded approximation ratio for a service chain embedding problem  that considers linear application graph and general physical graphs, and has both node and link capacity constraints. Therefore, similar to related work \cite{refViNEYard,PODC2011,PathsIntoTrees,AllocDistributedCloud,OptApproxLatencyMin,serviceChainInfocom2016,rost2016service,Lukovszki2015}, we make a few simplifications to make the problem tractable. These simplifications are driven by realistic MEC settings and their motivations are described as follows.

Throughout this paper, we focus on application and physical graphs that have tree topologies. 
This is due to the consideration that a tree application
graph models a wide range of MEC applications that involve a hierarchical set of processes (or virtual machines), including streaming, multicasting, and data aggregation applications
\cite{ChuMulticast,DataAggregationImpactWSN,ConcealedDataAggregation} such as the exemplar face recognition application presented earlier. 
For the physical system, we consider tree physical graphs due to the hierarchical nature of MEC environment (see Fig. \ref{chap:intro:fig:scenario}).
We note that the algorithms we propose in this paper \emph{also works with several classes of non-tree graphs} and an example will be given in Section~\ref{sec:discussion}. For ease of presentation, we mainly focus on tree graphs in this paper.

In the tree application graph, if we consider
any path from the root to a leaf, we only allow those assignments\footnote{We interchangeably use the terms ``placement'', ``assignment'', and
``mapping'' in this paper.} where the application nodes along this
path are assigned in their respective order on a sub-path
of the physical topology (multiple application nodes may still be placed onto one physical node), thus, creating a ``cycle-free''
placement. Figure~\ref{fig:cycles} illustrates this placement. Let nodes 1 to 5 denote the application nodes
along a path in the application-graph topology. The cycle-free placement
of this path onto a sub-path of the physical network ensures the order is preserved  (as shown in
Fig.~\ref{fig:cycles}(b)), whereas
the order is not preserved in Fig.~\ref{fig:cycles}(c). A cycle-free
placement has a clear motivation of avoiding cyclic communication
among the application nodes. For example, for the placement in Fig.~\ref{fig:cycles}(c),
application nodes 2 and 4 are placed on physical node B, while application
node 3 is placed on physical node C. In this case, the physical
link B--C carries the data of application links 2--3 and 3--4 in
a circular fashion. Such traffic can be naturally avoided with a cycle-free
mapping (Fig.~\ref{fig:cycles}(b)), thus relieving congestion on the communication links.
As we will see in the simulations in Section \ref{sec:SimulationResults}, the cycle-free constraint still allows the proposed scheme to outperform some other comparable schemes that allow cycles.
Further discussion on the approximation ratio associated with cycle-free restriction is given in Appendix \ref{app:ApproxRatioCycleFree}.

\begin{figure}
\center{\includegraphics[width=0.8\columnwidth]{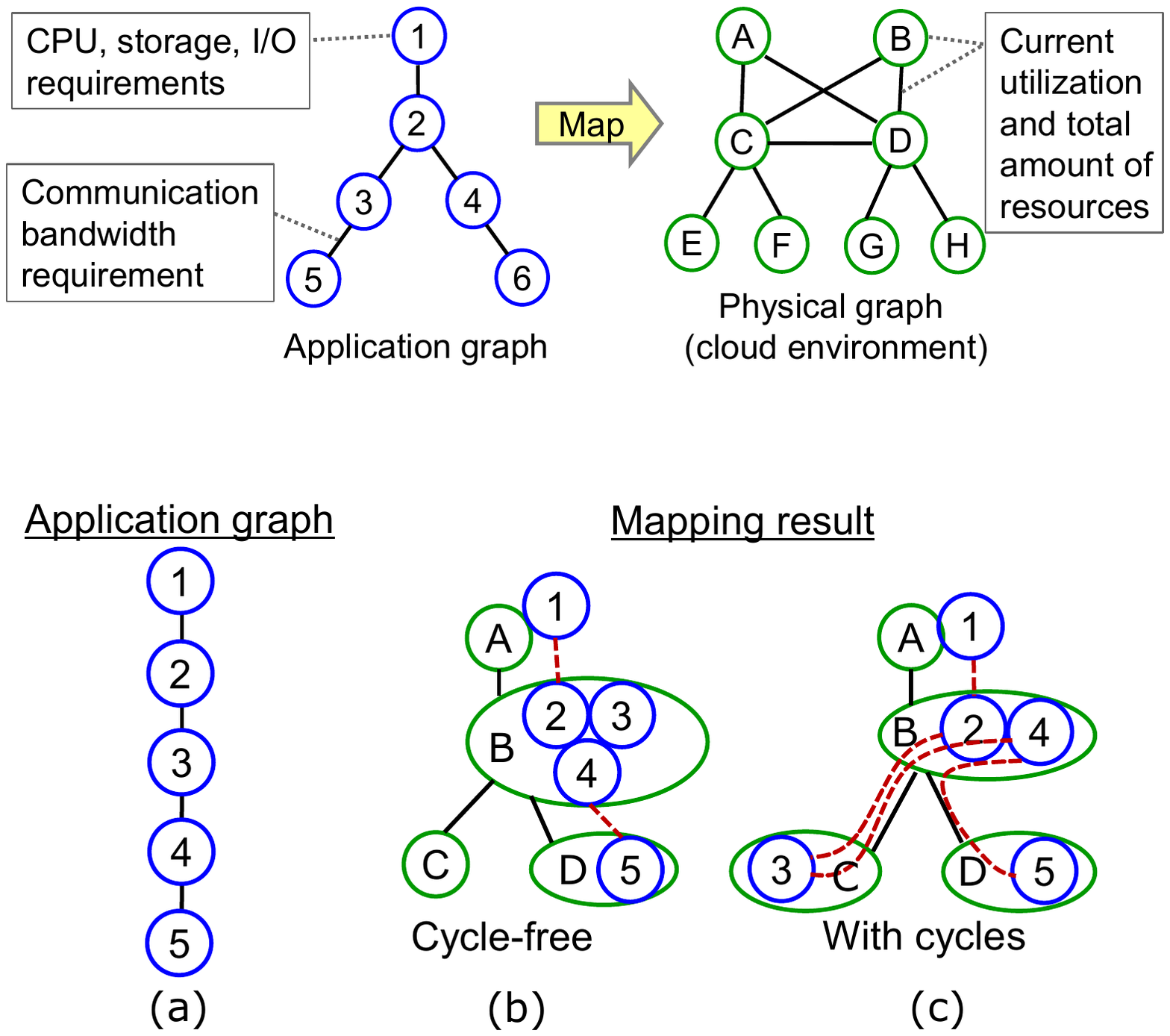}}

\caption{Mapping with and without cycles. In this example, the path in the
application graph is between application node 1 and application node
5.}

\label{fig:cycles} 

\end{figure}

In this paper, for the purpose of describing the algorithms, we classify an application node as a \emph{junction node} in the tree application graph when it has two or more children.
These junction nodes may represent data splitting or joining processes for multiple data streams. In some cases, they may have pre-specified placement, because they serve multiple data streams that may be associated with different end-users, and individual data streams may arrive dynamically in an online fashion. Our work first considers cases where the placements of all junction nodes (if any) are pre-specified, and then extends the results to the general case where some junction nodes are not placed beforehand. A linear application graph (such as the exemplar face recognition application in Section \ref{subsec:MotivatingExample}) has no junction nodes and it falls into the first category.

For the aforementioned scenario, we obtain the following main results
for the problem of application placement with the goal of load balancing
among physical nodes and edges: 
\begin{enumerate}
\item An optimal offline algorithm for placing a single application graph which is a linear graph, with $O(V{}^{3}N{}^{2})$ time-complexity and $O(VN(V+N))$ space-complexity, where the application graph has $V$ nodes and the physical graph has $N$ nodes. 
\item An online approximation algorithm for placing single or multiple tree application
graphs, in which the placements of all junction nodes are pre-specified, i.e., their placements are given. This algorithm has a time-complexity
of $O(V{}^{3}N{}^{2})$ and a space-complexity of $O(VN(V+N))$ for each application graph placement; its
competitive ratio is $O(\log N)$. 
\item An online approximation algorithm for placing single or multiple tree application
graphs, in which the placements of some junction nodes are \emph{not} pre-specified. This algorithm has a time-complexity of $O(V{}^{3}N{}^{2+H})$ and a space-complexity of $O(VN^{1+H}(V+N))$ for each application graph placement; its competitive ratio is
$O(\log^{1+H}N)$, where $H$ is the maximum number of 
junction nodes without given placement on any single path from the root to
a leaf in the application graph. Note that we always have $H\leq D$,
where $D$ is the depth of the tree application graph. 
\end{enumerate}

Our work considers multiple types of resources on each physical
node, such as CPU, storage, and I/O resources. The
proposed algorithms can work with domain constraints which restrict
the set of physical nodes that a particular application node can be
assigned to. The exact algorithm for single line placement can also
incorporate conflict constraints where some assignments are
not allowed for a pair of adjacent application nodes that are connected by an application edge; such constraints
may arise in practice due to security policies  as discussed in
\cite{VNEmbeddingSurvey}.

The remainder of this paper is organized in the following: Section \ref{sec:ProblemFormulation} describes the problem formulation. Section \ref{sec:singleLinePlacement} presents the exact optimal placement algorithm for single linear application graph. Online approximation algorithms for tree-to-tree placement are discussed in Section \ref{sec:onlinealgo}. Section \ref{sec:SimulationResults} shows numerical results. Some insights and observations are discussed in Section \ref{sec:discussion}. Section \ref{sec:conclusion} draws conclusions.

\section{Problem Formulation}

\label{sec:ProblemFormulation}

\subsection{Definitions}

We consider the placement of application graphs onto a physical graph,
where the application graphs represent applications that may arrive in
an online manner. In the following, we introduce some notations that will be used in this paper.

\textbf{Application Graph:} An application
is abstracted as a graph, in which nodes represent the processing/computational modules of the application (such as virtual machines or containers containing running application processes),
and edges represent the communication demand between nodes (such as information sharing among different application processes).
Each node $v\in\mathcal{{V}}$
in the application graph $\mathcal{{R}}=(\mathcal{{V}},\mathcal{{E}})$ is associated with parameters that represent
the computational resource (of $K$ different types) demands of node
$v$. Similarly, each edge $e\in\mathcal{{E}}$ is associated with
a communication bandwidth demand.  The notation $e=(v_{1},v_{2})$
denotes that application edge $e$ connects application nodes $v_{1}$ and $v_{2}$.
The application graph $\mathcal{{R}}$ can be either a directed or an undirected graph. If it is a directed graph, the direction of edges specify the direction of data communication; if it is an undirected graph, data communication can occur in either direction along application edges.

\textbf{Physical Graph:} The physical computing system is also abstracted as a graph, with nodes denoting computing devices\footnote{Multiple individual servers can be seen as a single entity if they constitute a single cloud.}
and edges denoting communication links between nodes.  
Each node $n\in\mathcal{{N}}$ in the physical graph  $\mathcal{{Y}}=(\mathcal{{N}},\mathcal{{L}})$
has $K$ different types of computational resources, and each
edge $l\in\mathcal{{L}}$ has communication resource. A physical node can also represent a network device such as a router or switch with zero computational resource.
We use the notation $l=(n_{1},n_{2})$ to denote that physical link $l$ connects physical nodes
$n_{1}$ and $n_{2}$.
Similar to the application graph, the physical graph can be either directed or undirected, depending on whether the physical links are bidirectional (i.e., communication in both directions share the same link) or single-directional (i.e., communication in each direction has a separate link).

Because we consider multiple application graphs in this paper, we denote the tree application graph for
the $i$th application arrival as $\mathcal{{R}}(i)=(\mathcal{{V}}(i),\mathcal{{E}}(i))$.
Throughout this paper, we define $V=|\mathcal{{V}}|$, $E=|\mathcal{{E}}|$, $N=|\mathcal{{N}}|$,
and $L=|\mathcal{{L}}|$, where $|\cdot|$ denotes the number of elements
in the corresponding set.

We consider undirected application and physical graphs in the problem formulation, which means that data can flow in any direction on an edge, but the proposed algorithms can be easily extended to many types of directed graphs. For example, when the tree application graph is directed and  the tree physical graph is undirected, we can merge the two application edges that share the same end nodes in different directions into one edge, and focus on the merged undirected application graph for the purpose of finding optimal placement. This does not affect the optimality because for any placement of application nodes, there is a unique path connecting two different application nodes due to the cycle-free constraint and the tree structure of physical graphs. Thus, application edges in both directions connecting the same pair of application nodes have to be placed along the same path on the physical graph.

\textbf{Costs:} 
For the $i$th application, the weighted cost (where the weighting factor can serve as a normalization to the total resource capacity) for type $k\in\{1,2,...,K\}$ resource of placing $v$ to $n$ is denoted by ${d}_{v\rightarrow n,k}(i)$. Similarly, the weighted communication bandwidth cost of assigning $e$ to $l$ is denoted by ${b}_{e\rightarrow l}(i)$. The edge cost is also defined for a dummy link $l=(n,n)$, namely a non-existing link that connects the same node, to take into account the additional cost when placing two application nodes on one physical node. It is also worth noting that an application edge may be placed onto multiple physical links that form a path.

\emph{Remark:} The cost of placing the same application node (or edge) onto different physical nodes (or edges) can be different. This is partly because different physical nodes and edges may have different resource capacities, and therefore different weighting factors for cost computation. It can also be due to the domain and conflict constraints as mentioned earlier. If some mapping is not allowed, then we can set the corresponding mapping cost to infinity. Hence, our cost definitions allow us to model a wide range of access-control/security policies.


\textbf{Mapping:} 
A mapping is specified by $\pi:\mathcal{{V}}\rightarrow\mathcal{{N}}$.
Because we consider
tree physical graphs with the cycle-free restriction, there exists only one path between two nodes
in the physical graph, and we use $(n_{1},n_{2})$ to denote either
the link or path between nodes $n_{1}$ and $n_{2}$. We use the notation
$l\in(n_{1},n_{2})$ to denote that link $l$ is included in
path $(n_{1},n_{2})$. The placement of nodes
automatically determines the placement of edges. 

In a successive placement of the $1$st up to the $i$th application, each physical node $n\in\mathcal{{N}}$
has an aggregated weighted cost~of 
\begin{equation}
{p}_{n,k}(i)=\sum_{j=1}^{i}\sum_{v:\pi(v)=n}{d}_{v\rightarrow n,k}(j),
\label{aggNodeCost}
\end{equation}
where the second sum is over all $v$ that are mapped to $n$. Equation (\ref{aggNodeCost}) gives the total cost of type $k$ resource requested by all application nodes that are placed on node $n$, upto the $i$th application. Similarly, each physical
edge $l\in\mathcal{{L}}$ has an aggregated weighted cost of
\begin{equation}
{q}_{l}(i)=\sum_{j=1}^{i}\sum_{e=(v_{1},v_{2}):\left(\pi(v_{1}),\pi(v_{2})\right)\ni l}{b}_{e\rightarrow l}(j),
\end{equation}
where the second sum is over all application edges $e=(v_{1},v_{2})$ for which the path between the physical nodes $\pi(v_{1})$ and $\pi(v_{2})$ (which $v_1$ and $v_2$ are respectively mapped to) includes the link $l$.

\subsection{Objective Function}
\label{subSec:overallObjectiveinProbFormulation}
The optimization objective in this paper is load balancing
for which the objective function is defined as
\begin{equation}
\min_{\pi}\max\left\{ \max_{k,n}{p}_{n,k}(M);\max_{l}{q}_{l}(M)\right\} ,\label{eq:objOnline}
\end{equation}
where $M$ is the total number of applications (application graphs). Equation (\ref{eq:objOnline}) aims to minimize the maximum weighted cost on each
physical node and link, ensuring that
no single element gets overloaded and becomes a point of failure, which is important especially in the presence of bursty traffic. Such an objective is widely used in the literature \cite{choudhury2014rejecting,AzarLoadBalancingSurvey}.

\emph{Remark:} While we choose the objective function (\ref{eq:objOnline}) in this paper, we do
realize that there can be other objectives as well, such as minimizing
the total resource consumption. We note that the exact algorithm for the placement of a single linear application graph \emph{can be generalized to a wide class of other objective functions}
as will be discussed in Section~\ref{sub:line-extensions}. For simplicity, we restrict our attention to the objective function in (\ref{eq:objOnline}) in most parts of our discussion.
We also note that our objective function incorporates both node and link resource consumptions, which is important in MEC environments where both node (server) and communication link conditions are closely related to the service performance.

\emph{A Note on Capacity Limit: }
For simplicity, we do not impose capacity constraints on physical nodes and links in the optimization problem defined in (\ref{eq:objOnline}), because even without the capacity constraint, the problem is very hard as we will see later in this paper. However, because the resource demand of each application node and link is specified in every application graph, the total resource consumption at a particular physical node/link can be calculated by summing up the resource demands of application nodes/links that are placed on it.
Therefore, an algorithm can easily check within polynomial time whether the current placement violates the capacity limits. If such a violation occurs, it can simply reject the newly arrived application graph.

In most practical cases, the costs of node and link placements should be defined as proportional to the resource occupation when performing such placement, with weights inversely proportional to the capacity of the particular type of resource. With such a definition, the objective function (\ref{eq:objOnline}) essentially tries to place as many application graphs as possible without increasing the maximum resource occupation (normalized by the resource capacity) among all physical nodes and links. Thus, the placement result should utilize the available resource reasonably well.
A more rigorous analysis on the impact of capacity limit is left as future work.

\section{Basic Assignment Unit: Single Linear Application Graph Placement}
\label{sec:singleLinePlacement}
We first consider
the assignment of a single linear application graph (i.e., the application nodes are connected in a line), where the goal is
to find the best placement of application nodes onto a path in the
tree physical graph under the cycle-free constraint (see
Fig.~\ref{fig:cycles}). The solution to this problem
forms the building block of other more sophisticated algorithms presented
later. As discussed next, we develop an algorithm that can find the optimal solution to
this problem. We omit the application index
$i$ in this section because we focus on a single application, i.e., $M=1$, here.

\subsection{Sub-Problem Formulation}
\label{subSec:lineToLineProbFormulation}

Without loss of generality, we assume that $\mathcal{{V}}$ and $\mathcal{{N}}$
are indexed sets, and we use $v$ to exchangeably denote elements
and indices of application nodes in $\mathcal{{V}}$, and use $n$ to  exchangeably denote elements
and indices of physical nodes in $\mathcal{{N}}$. This index (starting from $1$ for the root node) is determined
by the topology of the graph. In particular, it can be determined
via a breadth-first or depth-first indexing on the tree graph (note that  linear graphs are a special type of tree graphs).
From this it follows that, if $n_{1}$ is a parent of $n_{2}$, then
we must have $n_{1}<n_{2}$. The same holds for the application nodes $\mathcal{{V}}$. 

With this setting, the edge cost can be combined together with the cycle-free constraint into a single definition of
pairwise costs. The weighted pairwise cost of placing $v-1$ to
$n_{1}$ and $v$ to $n_{2}$ is denoted by $c_{(v-1,v)\rightarrow(n_{1},n_{2})}$,
and it takes the following values with $v\geq2$:
\begin{itemize}
\item[--] If the path from $n_{1}$ to $n_{2}$ traverses some $n<n_{1}$, in which case the cycle-free assumption is violated, then $c_{(v-1,v)\rightarrow(n_{1},n_{2})}=\infty$. 
\item[--] Otherwise,
\begin{equation}
c_{(v-1,v)\rightarrow(n_{1},n_{2})}=\max_{l\in(n_{1},n_{2})} {b}_{(v-1,v)\rightarrow l} \Big{|}_{\left(\pi(v-1),\pi(v)\right)\ni l}.\label{eq:pairwiseCost}
\end{equation}
\end{itemize}
The maximum operator in (\ref{eq:pairwiseCost}) follows from the fact that, in the single line
placement, at most one application edge can be placed onto a physical
link. Also recall that the edge cost definition incorporates dummy links such as $l=(n,n)$, thus there always exists $l\in(n_{1},n_{2})$ even if $n_{1}=n_{2}$.

Then, the optimization problem (\ref{eq:objOnline}) with $M=1$ becomes
\begin{align}
\min_{\pi}\max  \Bigg\{ & \max_{k,n}\sum_{v:\pi(v)=n}{d}_{v\rightarrow n,k}; \nonumber \\
& \max_{(v-1,v)\in\mathcal{{E}}}c_{(v-1,v)\rightarrow(\pi(v-1),\pi(v))}\Bigg\}. \label{eq:objOfflineWithPairwise}
\end{align}
The last maximum operator in (\ref{eq:objOfflineWithPairwise}) takes
the maximum among all application edges (rather than physical links),
because when combined with the maximum in (\ref{eq:pairwiseCost}),
it essentially computes the maximum among all physical links that are used for data transmission under the mapping $\pi$.

\subsection{Decomposing the Objective Function}

In this subsection, we decompose the objective function in (\ref{eq:objOfflineWithPairwise})
to obtain an iterative solution. Note that the objective function
(\ref{eq:objOfflineWithPairwise}) already incorporates all the constraints as discussed earlier. Hence, we only
need to focus on the objective function itself.

When only considering a subset of application nodes $1,2,...,v_{1}\leq V$,
for a given mapping $\pi$, the value of the objective function for
this subset of application nodes is 
\begin{align}
J_{\pi}(v_{1}) = \max\Bigg\{ & \max_{k,n}\sum_{v\leq v_{1}:\pi(v)=n}{d}_{v\rightarrow n,k}; \nonumber \\
& \max_{(v-1,v)\in\mathcal{{E}},v\leq v_{1}}c_{(v-1,v)\rightarrow(\pi(v-1),\pi(v))}\Bigg\} . 
\label{eq:objOfflineWithPairwiseAppSubset}
\end{align}
Compared with (\ref{eq:objOfflineWithPairwise}), the only difference
in (\ref{eq:objOfflineWithPairwiseAppSubset}) is that we consider
the first $v_{1}$ application nodes and the mapping $\pi$ is assumed
to be given. The optimal cost for application nodes $1,2,...,v_{1}\leq V$
is then 
\begin{equation}
J_{\pi^{*}}(v_{1})=\min_{\pi}J_{\pi}(v_{1}),\label{eq:opt-objOfflineWithPairwiseAppSubset}
\end{equation}
where $\pi^{*}$ denotes the optimal mapping.

{\proposition \label{prop:singleline-objfun-decomp}(\textbf{Decomposition
of the Objective Function}): Let $J_{\pi^{*}|\pi(v_{1})}(v_{1})$
denote the optimal cost under the condition that $\pi(v_{1})$ is given,
i.e. $J_{\pi^{*}|\pi(v_{1})}(v_{1})=\min_{\pi(1),...,\pi(v_{1}-1)}J_{\pi}(v_{1})$ with given $\pi(v_{1})$. When $\pi(v_{1})=\pi(v_{1}-1)=...=\pi(v{}_{s})>\pi(v{}_{s}-1)\geq\pi(v{}_{s}-2)\geq...\geq\pi(1)$,
where $1\leq v_{s}\leq v_{1}$, which means that $v{}_{s}$ is mapped
to a different physical node from $v{}_{s}-1$ and nodes $v_{s},...,v_{1}$
are mapped onto the same physical node%
\footnote{Note that when $v_{s}=1$, then $v_{s}-1$ does not exist, which means
that all nodes $1,...,v_{1}$ are placed onto the same physical node.
For convenience, we define $J_{\pi}(0)=0$.%
}, then we have 
\begin{eqnarray}
J_{\pi^{*}|\pi(v_{1})}(v_{1}) \!\!\!\!\!& = &\!\!\!\!\! \min_{v{}_{s}=1,...,v_{1}}\min_{\pi(v_{s}-1)}\max\bigg\{ J_{\pi^{*}|\pi(v{}_{s}-1)}(v_{s}-1);\nonumber \\
 &  &\!\!\!\!\! \max_{k=1,...,K}\sum_{v=v{}_{s}...v_{1}}{d}_{v\rightarrow\pi(v_1),k};\nonumber \\
 &  &\!\!\!\!\!\!\!\!\!\! \left.\max_{(v-1,v)\in\mathcal{{E}},v_{s}\leq v\leq v_{1}}c_{(v-1,v)\rightarrow(\pi(v-1),\pi(v))}\right\} .\label{eq:iterativeFunction}
\end{eqnarray}
The optimal mapping for $v_1$ can be found by
\begin{equation}
J_{\pi^{*}}(v_{1})=\min_{\pi(v_{1})}J_{\pi^{*}|\pi(v_{1})}(v_{1}).\label{eq:iterativeFunctionFinal}
\end{equation}
}
\begin{proof}
Because $\pi(v_s)=\pi(v_s+1)=...=\pi(v_1)$, we have
\begin{eqnarray}
J_{\pi}(v_{1})\!\!\! & = & \!\!\!\max\left\{ J_{\pi}(v_{s}-1);\max_{k=1,...,K}\sum_{v=v{}_{s}...v_{1}}{d}_{v\rightarrow\pi(v_{1}),k};\right.\nonumber \\
 &  & \left.\max_{(v-1,v)\in\mathcal{{E}},v_{s}\leq v\leq v_{1}}c_{(v-1,v)\rightarrow(\pi(v-1),\pi(v))}\right\}\!\!.\label{eq:objDecomp}
\end{eqnarray}
The three terms in the maximum operation in (\ref{eq:objDecomp})
respectively correspond to: 1) the cost at physical nodes and edges
that the application nodes $1,...,v_{s}-1$ (and their connecting
edges) are mapped to, 2) the costs at the physical node that $v_{s},...,v_{1}$
are mapped to, and 3) the pairwise costs for connecting $v_{s}-1$
and $v_{s}$ as well as interconnections\footnote{Note that, although $v_{s},...,v_{1}$ are mapped onto the same physical node, their pairwise costs may be non-zero if there exists additional cost when placing different application nodes onto the same physical node. In the extreme case where adjacent application nodes are not allowed to be placed onto the same physical node (i.e., conflict constraint), their pairwise cost when placing them on the same physical node becomes infinity.} of nodes in $v_{s},...,v_{1}$. 
Taking the maximum of these three terms, we obtain the cost function
in (\ref{eq:objOfflineWithPairwiseAppSubset}).

In the following, we focus on finding the optimal mapping based on
the cost decomposition in (\ref{eq:objDecomp}). We note that the
pairwise cost between $v_{s}-1$ and $v_{s}$ depends on the placements
of both $v_{s}-1$ and $v_{s}$. Therefore, in order to find the optimal
$J_{\pi}(v_{1})$ from $J_{\pi}(v_{s}-1)$, we need to find the minimum cost among all possible placements of $v_{s}-1$ and $v_{s}$, provided that nodes $v_{s},...,v_{1}$
are mapped onto the same physical node and $v_{s}$ and $v_{s}-1$
are mapped onto different physical nodes. For a given $v_{1}$,
node $v_{s}$ may be any node that satisfies $1\leq v_{s}\leq v_{1}$.
Therefore, we also need to search through all possible values of $v_{s}$. This can then be expressed as the proposition, where we first find $J_{\pi^{*}|\pi(v_{1})}(v_{1})$ as an intermediate step.
\end{proof}

Equation (\ref{eq:iterativeFunction}) is the Bellman's equation \cite{powell2007approximate} for problem (\ref{eq:objOfflineWithPairwise}). Using dynamic programming \cite{powell2007approximate}, we can solve  (\ref{eq:objOfflineWithPairwise}) by iteratively
solving (\ref{eq:iterativeFunction}). In each iteration, the algorithm computes new
costs $J_{\pi^{*}|\pi(v_{1})}(v_{1})$ for all possible mappings $\pi(v_{1})$,
based on the previously computed costs $J_{\pi^{*}|\pi(v)}(v)$ where
$v<v_{1}$. For the final application node $v_{1}=V$, we use (\ref{eq:iterativeFunctionFinal})
to compute the final optimal cost $J_{\pi^{*}}(V)$ and its corresponding
mapping $\pi^{*}$.

\subsection{Optimal Algorithm}
\label{sub:LineToTreeAlgorithm} 
The pseudocode of the exact optimal
algorithm is shown in Algorithm \ref{algLineToTree}.
It computes $V\cdot N$ number of $J_{\pi^{*}|\pi(v)=n}(v)$
values, and we take the minimum
among no more than $V\cdot N$ values in (\ref{eq:iterativeFunction}). The terms in (\ref{eq:iterativeFunction})
include the sum or maximum of no more than $V$ values and the maximum
of $K$ values. Because $K$ is a constant in practical
systems, we conclude that the \emph{time-complexity of this algorithm is $O(V{}^{3}N{}^{2})$}.

The  \emph{space-complexity of Algorithm~\ref{algDoubling} is  $O(VN(V+N))$}, which is related to the memory required for storing matrices  $\mathbf{D}$, $\mathbf{C}$, $\mathbf{J}$, and $\mathbf{\Pi}$ in the algorithm, where $K$ is also regarded as a constant here.

Also note that the optimality of the result from Algorithm~\ref{algLineToTree} is subject to the cycle-free constraint, and the sequence of nodes is always preserved in each iteration.

\begin{algorithm}
\caption{Placement of a linear application graph onto a tree physical graph}
\label{algLineToTree}
{
\begin{algorithmic}[1]
\STATE
Given linear application graph $\mathcal{{R}}$, 
tree physical graph $\mathcal{{Y}}$
\STATE
Given $V\times N\times K$ matrix $\mathbf{D}$, its
entries represent the weighted type $k$ node cost ${d}_{v\rightarrow n,k}$
\STATE
Given $(V-1)\times N\times N$ matrix $\mathbf{C}$,
its entries represent the weighted pairwise cost $c_{(v-1,v)\rightarrow(n_{1},n_{2})}$
\STATE
Define $V\times N$ matrix $\mathbf{J}$ to keep the
costs $J_{\pi^{*}|\pi(v)=n}(v)$ for each node $(v,n)$ in the auxiliary graph 
\STATE
Define $V\times N\times V$ matrix $\mathbf{\Pi}$ to
keep the mapping corresponding to its cost $J_{\pi^{*}|\pi(v)=n}(v)$ for each node $(v,n)$ in the auxiliary graph
\FOR {$v=1...V$}
\FOR {$n=1...N$}
\STATE
Compute $J_{\pi^{*}|\pi(v)=n}(v)$ from (\ref{eq:iterativeFunction}),
put the result into $\mathbf{J}$ and the corresponding mapping
into $\mathbf{\Pi}$
\ENDFOR
\ENDFOR
\STATE
Compute $J_{\pi^{*}}(V) \leftarrow \min_{n}J_{\pi^{*}|\pi(V)=n}(V)$
\STATE \textbf{return} the final mapping result $\pi^{*}$ and final optimal cost $J_{\pi^{*}}(V)$ 
\end{algorithmic}
}
\end{algorithm}

\subsection{Example}

\label{sub:line-example}

\begin{figure}
\center{\includegraphics[width=1\columnwidth]{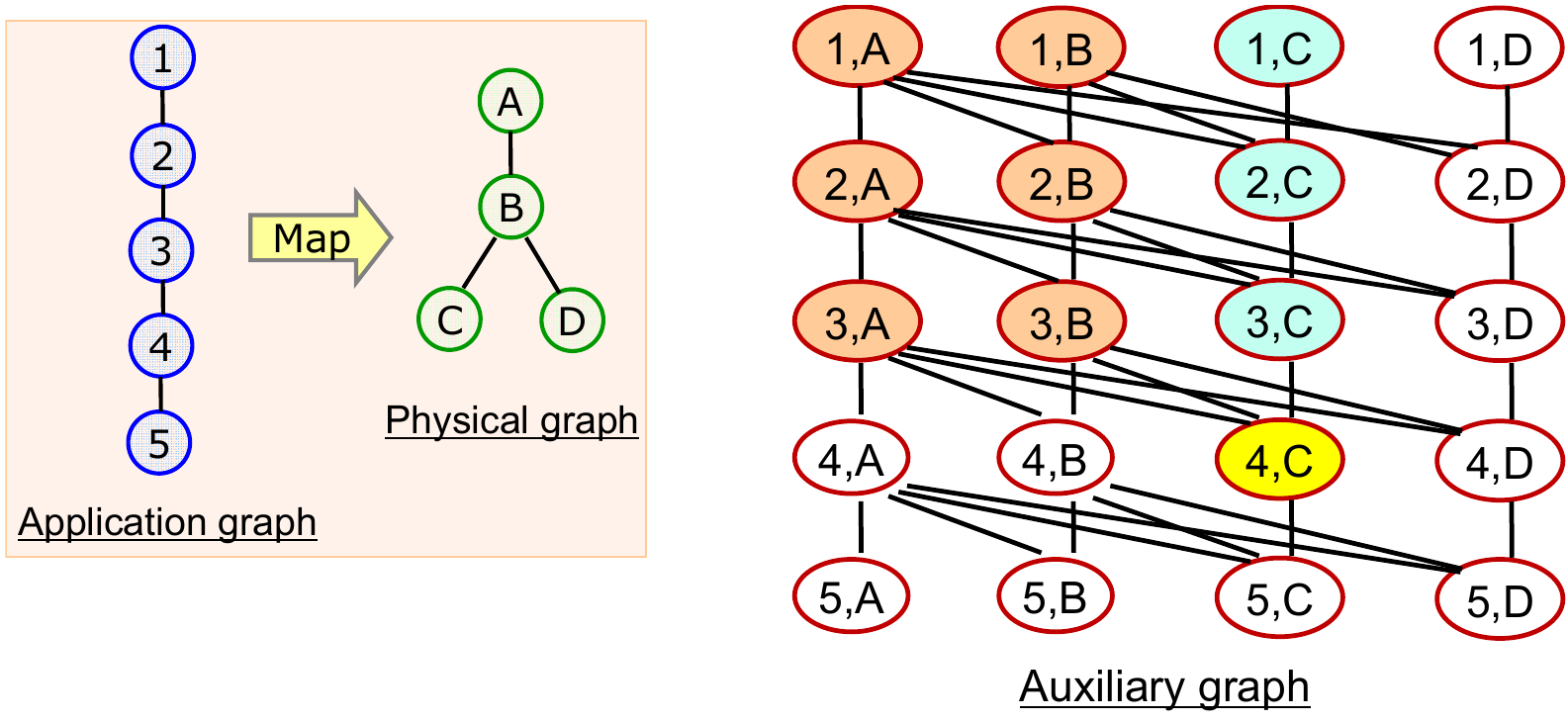}}
\caption{Auxiliary graph and algorithm procedure for the placement of a linear application graph onto a tree physical graph.}

\label{fig:auxGraph} 

\end{figure}

To illustrate the procedure of the algorithm, we construct an auxiliary
graph from the given application and physical graphs, as shown in
Fig. \ref{fig:auxGraph}. Each node $(v_1,n_1)$ in the auxiliary graph represents
a possible placement of a particular application node, and is associated with
the cost value $J_{\pi^{*}|\pi(v_{1})=n_1}(v_{1})$, where $v_{1}$ is
the application node index and $n_1$ is the physical node
index in the auxiliary graph. When computing the cost at a particular
node, e.g. the cost $J_{\pi^{*}|\pi(4)=\textrm{C}}(4)$
at node ($4$,C) in Fig. \ref{fig:auxGraph},
the algorithm starts from the ``earlier'' costs $J_{\pi^{*}|\pi(v{}_{s}-1)}(v_{s}-1)$
where the tuple $(v_{s}-1,\pi(v_{s}-1))$ is either ($1$,A),  ($1$,B), ($2$,A), ($2$,B), ($3$,A), or
($3$,B). From each of these nodes,
the subsequent application nodes (i.e. from $v_s$ to node $4$) are all mapped onto physical node
C, and we compute the cost for
each such ``path'' with the maximum operations in (\ref{eq:iterativeFunction}), by assuming the
values of $v_{s}-1$ and $\pi(v_{s}-1)$ are given by its originating
node in the auxiliary graph. For example, one path can be ($2$,B) -- ($3$,C) -- ($4$,C) where $v_{s}-1=2$ and $\pi(v_{s}-1)$ = 
B, another path can be ($1$,A) -- ($2$,C) -- ($3$,C) -- ($4$,C) where
$v_{s}-1=1$ and $\pi(v_{s}-1)$ = A. Then, the algorithm takes the minimum of the costs for all paths,
which corresponds to the minimum operations in (\ref{eq:iterativeFunction}) and gives $J_{\pi^{*}|\pi(4)=\textrm{C}}(4)$.
In the end, the algorithm searches through all the possible mappings
for the final application node (node $5$ in Fig. \ref{fig:auxGraph})
and chooses the mapping that results in the minimum cost, which corresponds to the procedure in (\ref{eq:iterativeFunctionFinal}).

\subsection{Extensions}

\label{sub:line-extensions}

The placement algorithm for single linear application graph can also be used when the objective function (in the form of (\ref{eq:objOnline}) with $M=1$) 
is modified to one of the following:
\begin{align}
\min_{\pi}\max  \Bigg\{ & \max_{k,n}f_{n,k}\Bigg(\sum_{v:\pi(v)=n}{d}_{v\rightarrow n,k}\Bigg);\nonumber \\
& \max_{l}g_{l}\Bigg(\sum_{e=(v_{1},v_{2}):\left(\pi(v_{1}),\pi(v_{2})\right)\ni l}{b}_{e\rightarrow l}\Bigg)\!\!\Bigg\},  \label{eq:objOfflineExtMax}
\end{align}
\begin{align}
\min_{\pi} & \Bigg\{ \sum_{k,n}f_{n,k}\Bigg(\sum_{v:\pi(v)=n}{d}_{v\rightarrow n,k}\Bigg) \nonumber \\
& + \sum_{l}g_{l}\Bigg(\sum_{e=(v_{1},v_{2}):\left(\pi(v_{1}),\pi(v_{2})\right)\ni l}{b}_{e\rightarrow l}\Bigg) \Bigg\} ,  
\label{eq:objOfflineExtSum}
\end{align}
where $f_{n,k}(\cdot)$ and $g_{l}(\cdot)$ are increasing functions with $f_{n,k}(0)=0$, $g_{l}(0)=0$, $f_{n,k}(\infty)=\infty$, and $g_l(\infty)=\infty$.
The algorithm and its derivation follow the same procedure as discussed
above. These alternative objective functions can be useful for scenarios where the goal of optimization is other than min-max. The objective function in (\ref{eq:objOfflineExtSum}) will also be used later for solving the online placement problem.

\section{Online Placement Algorithms for Tree Application Graphs}

\label{sec:onlinealgo} 

Using the optimal algorithm for the single
linear application graph placement as a sub-routine, we now present
algorithms for the generalized case; namely, placement of an arriving
stream of application graphs with tree topology. 
We first show that even the offline placement of a single tree is
NP-hard. Then, we propose online algorithms to approximate the optimal
placement with provable competitive ratio, by first
considering the case where junction nodes in the application
graph have pre-specified placements that are given beforehand, and later relax this assumption.

\subsection{Hardness Result}
\label{sec:hardnessTreeToTree}

\begin{proposition}\label{propTreeToTreeNPHard} (\textbf{NP-hardness})
Placement of a tree application graph onto a tree physical graph for
the objective function defined in (\ref{eq:objOnline}), with or without pre-specified junction node placement, is NP-hard.
\end{proposition}
\begin{proof} 
To show that the given problem is NP-hard,
we show that the problem can be reduced from the NP-hard problem of minimum makespan
scheduling on unrelated parallel machines (MMSUPM) \cite{bookApproxAlg},
which minimizes the maximum load (or job processing time) on each
machine.

Consider a special case in our problem where the application graph
has a star topology with two levels (one root and multiple leaf nodes), and the physical graph
is a line with multiple nodes. Assume that
the number of resource types in the nodes is $K=1$, the application
edge resource demand is zero, and the application node resource demand
is non-zero. Then, the problem is essentially the MMSUPM problem.
It follows that the MMSUPM problem reduces to our problem. In other
words, if we can solve our problem in polynomial time, then we can
also solve the MMSUPM problem in polynomial time. Because MMSUPM is NP-hard, our problem is also NP-hard. The above result holds no matter whether the root node (junction node) of the application graph has pre-specified placement or not.
\end{proof}

\subsection{When All Junction Node Placements Are Given}

\begin{figure}
\center{\includegraphics[width=0.8\columnwidth]{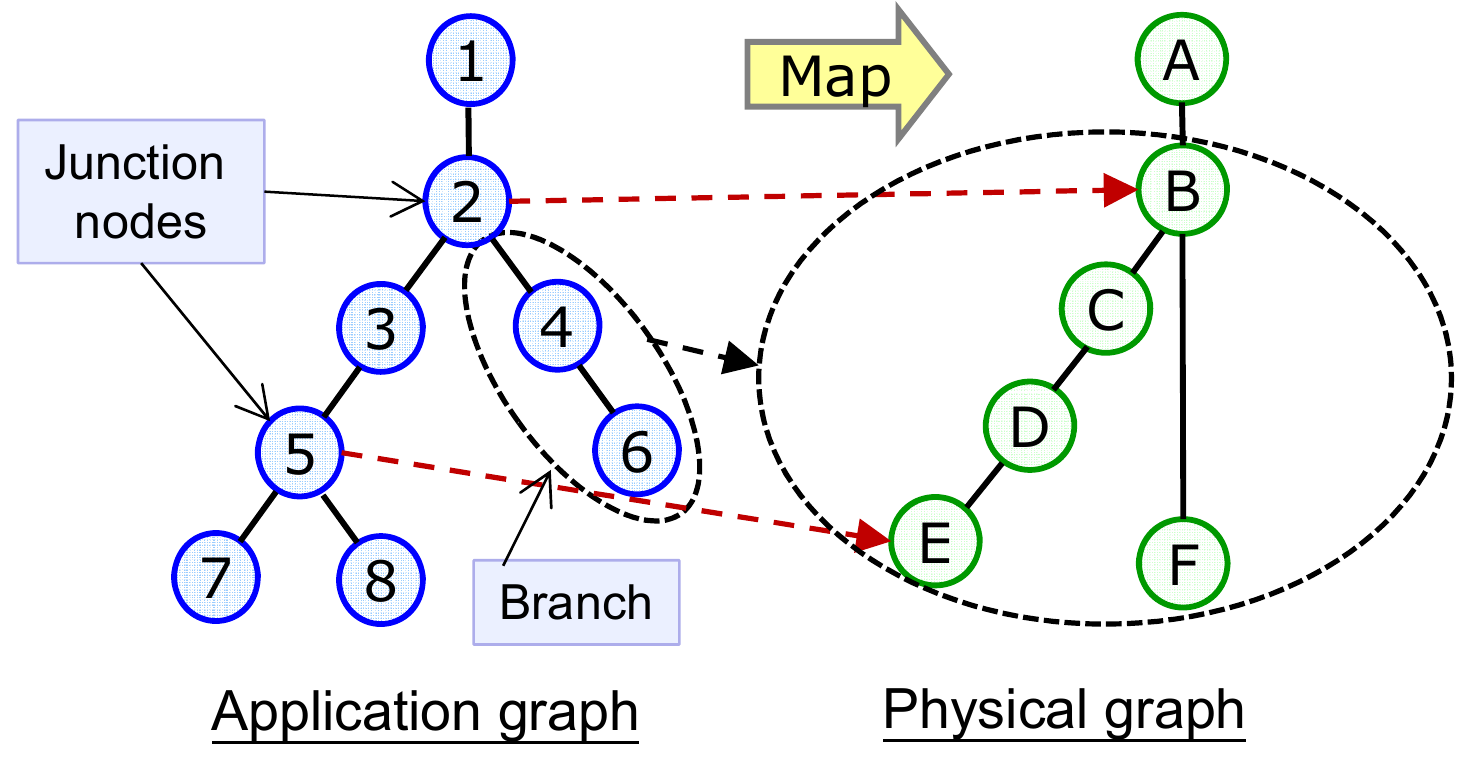}}
\caption{Example of application graph with given placement of junction nodes. Junction node $2$ is placed on physical node B  
and junction
node $5$ is placed on physical node E. The algorithm needs to decide
the placement of the remaining nodes, subject to the cycle-free constraint. }
\label{fig:exampleFixedDegree3} 
\end{figure}

\label{sec:FixedDegree3} We first consider tree application
graphs for which the placements of junction nodes are given, and focus on placing the remaining
non-junction nodes which are connected to at most two edges. An example is shown
in Fig. \ref{fig:exampleFixedDegree3}. Given the placed junction nodes,
we name the set of application edges and nodes that form a chain between
the placed nodes (excluding each placed node itself, but including 
each edge that is connected to a placed node) as a \emph{simple branch}, where the notion ``simple'' is opposed to the general branch which will be defined in Section \ref{sec:UnfixedDegree3}. 
A simple branch can also be a chain starting from an edge that connects a placed node
and ending at a leaf node, such as the nodes and edges within the
dashed boundary in the application graph in Fig. \ref{fig:exampleFixedDegree3}.
Each node in a simple branch is connected to at most two edges.

\subsubsection{Algorithm Design}
\label{sub:unplacedAlgDesign}
We propose an online placement algorithm, where we borrow some ideas from \cite{OnlineRoutingWithExpCost}. Different from \cite{OnlineRoutingWithExpCost}
which focused on routing and job scheduling problems, our work considers more general graph mapping.

When an application (represented by a tree application graph) arrives, we
\emph{split} the whole application graph into simple branches, and regard each simple branch
as an independent application graph. All the nodes with given placement can also be regarded
as an application that is placed before placing the individual simple branches.
After placing those nodes, each individual simple branch is placed using
the online algorithm that we describe below. In the remaining
of this section, by application we refer to the \emph{application after
splitting}, i.e. each application either consists of a simple branch or a set of nodes with given placement.

\textbf{How to Place Each Simple Branch: } While our ultimate goal is to optimize (\ref{eq:objOnline}), we use an \emph{alternative objective function} to determine the placement of each newly arrived application $i$ (after splitting). Such an indirect approach provides performance guarantee with respect to (\ref{eq:objOnline}) in the long run. We will first introduce the new objective function and then discuss its relationship with the original objective function (\ref{eq:objOnline}). 

We define a variable $\hat{J}$ as a reference cost. The reference
cost may be an estimate of the true optimal cost (defined as in (\ref{eq:objOnline})) from
optimal offline placement, and we will discuss later about how to determine this value. 
Then, for placing the $i$th application, we use an objective function which has the same form as (\ref{eq:objOfflineExtSum}), with $f_{n,k}(\cdot)$ and $g_{l}(\cdot)$ defined as
\begin{subequations}
\begin{align}
&f_{n,k}(x)     \triangleq     \exp_\alpha \! \left( {\frac{{p}_{n,k}(i-1)+x}{\hat{J}}} \right)\! - \!\exp_\alpha \! \left( {\frac{{p}_{n,k}(i-1)}{\hat{J}}} \right) \label{eq:onlineExpObj_fnk}  , \\
&g_{l}(x)    \triangleq      \exp_\alpha\!\left({\frac{{p}_{l}(i-1)+x}{\hat{J}}}\right)\!-\!\exp_\alpha\!\left({\frac{{p}_{l}(i-1)}{\hat{J}}}\right)
\label{eq:onlineExpObj_fl}  ,
\end{align}
\end{subequations}
subject to the cycle-free placement constraint, where we define $\exp_\alpha(y)\triangleq \alpha^y$ and $\alpha\triangleq 1+1/\gamma$ ($\gamma>1$ is a design parameter). 

\emph{Why We Use an Alternative Objective Function: }
The objective function (\ref{eq:objOfflineExtSum}) with (\ref{eq:onlineExpObj_fnk}) and (\ref{eq:onlineExpObj_fl}) is the increment of the sum exponential values of the original costs, given all the previous placements. With this objective function, the performance bound of the algorithm can be shown analytically (see Proposition \ref{prop:Bound} below). Intuitively, the new objective function (\ref{eq:objOfflineExtSum}) serves the following purposes:
\begin{itemize}
\item[--] ``Guide'' the system into a state such that the maximum cost among physical links and nodes is not too high, thus approximating the original objective function (\ref{eq:objOnline}). This is because when the existing cost at a physical link or node (for a particular resource type $k$) is high, the incremental cost (following (\ref{eq:objOfflineExtSum})) of placing the new application $i$ on this link or node (for the same resource type $k$) is also high, due to the fact that $\exp_\alpha(y)$ is convex increasing and the cost definitions in (\ref{eq:onlineExpObj_fnk}) and (\ref{eq:onlineExpObj_fl}). 

\item[--] While (\ref{eq:objOnline}) only considers the maximum cost, (\ref{eq:objOfflineExtSum}) is also related to the sum cost, because we sum up all the exponential cost values at different physical nodes and links together. This ``encourages'' a low resource consuming placement of the new application $i$ (which is reflected by low sum cost), thus leaving more available resources for future applications. In contrast, if we use (\ref{eq:objOnline}) directly for each newly arrived application, the placement may greedily take up too much resource, so that future applications can no longer be placed with a low cost. 
\end{itemize}
In practice, we envision that objective functions with a shape similar  to (\ref{eq:objOfflineExtSum}) can also serve  our purpose.

\emph{How to Solve It: }
Because each application either obeys a pre-specified placement or consists of a simple branch, we can use Algorithm~\ref{algLineToTree} with appropriately modified
cost functions to find the optimal solution to (\ref{eq:objOfflineExtSum}) with (\ref{eq:onlineExpObj_fnk}) and (\ref{eq:onlineExpObj_fl}).
For the case of a simple branch having an open edge, such as edge $(2,4)$ in
Fig. \ref{fig:exampleFixedDegree3}, we connect an application node
that has zero resource demand to extend the simple branch to a graph, so that Algorithm \ref{algLineToTree} is applicable.

\begin{algorithm}
\caption{Online placement of an application that is either a simple branch or a set of nodes with given placement}
\label{algMultiLine}
{
\begin{algorithmic}[1]
\STATE
Given the $i$th application that is either a set of nodes with given
placement or a simple branch
\STATE
Given tree physical graph $\mathcal{{Y}}$
\STATE
Given ${p}{}_{n,k}(i-1)$, ${q}{}_{l}(i-1)$, and placement
costs
\STATE
Given $\hat{J}$ and $\beta$
\IF {application is a set of nodes with given placement}
\STATE 
Obtain $\pi_{i}$ based on given placement
 
\ELSIF {application is a simple branch}
\STATE
Extend simple branch to linear graph $\mathcal{{R}}(i)$, by connecting zero-resource-demand application nodes to open edges,
and the placements of these zero-resource-demand application nodes are given
\STATE Run Algorithm \ref{algLineToTree} with objective function
(\ref{eq:objOfflineExtSum}) with (\ref{eq:onlineExpObj_fnk}) and (\ref{eq:onlineExpObj_fl}), for $\mathcal{{R}}(i)$, to obtain $\pi_{i}$
\ENDIF

\IF {$\exists n,k:{p}_{n,k}(i-1)+\sum_{v:\pi_{i}(v)=n}{d}_{v\rightarrow n,k}(i)>\beta\hat{J}$
or $\exists l:{q}_{l}(i-1)+\sum_{e=(v_{1},v_{2}):\left(\pi_{i}(v_{1}),\pi_{i}(v_{2})\right)\ni l}{b}_{e\rightarrow l}(i)>\beta\hat{J}$ }

\STATE \textbf{return} FAIL
\ELSE
\STATE \textbf{return} $\pi_{i}$
\ENDIF
\end{algorithmic}
}
\end{algorithm}

Algorithm \ref{algMultiLine} summarizes the above argument
as a formal algorithm for each application placement, where $\pi_{i}$ denotes the mapping for the $i$th application. Define a
parameter,  $\beta=\log_{\alpha}\left(\frac{\gamma(NK+L)}{\gamma-1}\right)$,
then Algorithm \ref{algMultiLine} performs the placement as long
as the cost on each node and link is not bigger than $\beta\hat{J}$,
otherwise it returns FAIL. The significance of the parameter $\beta$
is in calculating the competitive ratio, i.e., the maximum ratio of the cost
resulting from Algorithm \ref{algMultiLine} to the optimal cost from
an equivalent offline placement, as shown below.

\emph{Why We Need the Reference Cost $\hat{J}$: }
The reference cost $\hat{J}$ is an input parameter of the objective function (\ref{eq:objOfflineExtSum}) and Algorithm \ref{algMultiLine}, which enables us to show a performance bound for Algorithm \ref{algMultiLine}, as shown in Proposition \ref{prop:Bound}.

\begin{proposition}\label{prop:Bound}If there exists an offline
mapping $\pi^{o}$ that considers all $M$ application graphs and brings cost
$J_{\pi^{o}}$, such that $J_{\pi^{o}}\leq\hat{J}$, then Algorithm
\ref{algMultiLine} never fails, i.e., ${p}_{n,k}(M)$ and ${q}_{l}(M)$
from Algorithm \ref{algMultiLine} never exceeds $\beta\hat{J}$.
The cost $J_{\pi^{o}}$ is defined in (\ref{eq:objOnline}).
\end{proposition}
\begin{proof} 
See Appendix \ref{app:ProofGraphPlacementApproxBound}.
\end{proof}

Proposition \ref{prop:Bound} guarantees a bound for the cost resulting
from Algorithm \ref{algMultiLine}. We note that the optimal offline
mapping $\pi^{o*}$ produces cost $J_{\pi^{o*}}$, which is smaller
than or equal to the cost of an arbitrary offline mapping. It follows that for any $\pi^{o}$, we have $J_{\pi^{o*}}\leq J_{\pi^{o}}$.
This means that if there exists $\pi^{o}$ such that $J_{\pi^{o}}\leq\hat{J}$,
then we must have $J_{\pi^{o*}}\leq\hat{J}$. If we can set $\hat{J}=J_{\pi^{o*}}$,
then from Proposition~\ref{prop:Bound} we have $\max\left\{ \max_{k,n}{p}_{n,k}(M);\max_{l}{q}_{l}(M)\right\} \leq\beta J_{\pi^{o*}}$,
which means that the competitive ratio is $\beta$.

\textbf{How to Determine the Reference Cost $\hat{J}$:} Because the value of $J_{\pi^{o*}}$ is unknown, we cannot
always set $\hat{J}$ exactly to $J_{\pi^{o*}}$. Instead, we need
to set $\hat{J}$ to an estimated value that is not too far from $J_{\pi^{o*}}$.
We achieve this by using the doubling technique, which is widely used
in online approximation algorithms. The idea 
is to double the value of $\hat{J}$ every time  Algorithm \ref{algMultiLine}
fails. After each doubling, we ignore all the previous placements when calculating the objective function
(\ref{eq:objOfflineExtSum}) with (\ref{eq:onlineExpObj_fnk}) and (\ref{eq:onlineExpObj_fl}),
i.e., we assume that there is no existing application,
and we place the subsequent applications (including the one that has
failed with previous value of $\hat{J}$) with the new value of $\hat{J}$.
At initialization, the value of $\hat{J}$ is set to a reasonably
small number $\hat{J}_{0}$.

In Algorithm \ref{algDoubling}, we summarize the high-level procedure
that includes the splitting of the application graph, the calling of Algorithm
\ref{algMultiLine}, and the doubling process, with multiple application
graphs that arrive over time.

\begin{algorithm}
\caption{High-level procedure for multiple arriving tree application graphs}
\label{algDoubling}
{
\begin{algorithmic}[1]
\STATE
Initialize $\hat{J}\leftarrow\hat{J}_{0}$
\STATE
Define index $i$ as the application index, which automatically increases
by 1 for each new application (after splitting)
\STATE
Initialize $i\leftarrow 1$
\STATE
Initialize $i_{0}\leftarrow 1$
\LOOP
\IF {new application graph has arrived}
\STATE
Split the application graph into simple branches and a
set of nodes with given placement, assume that each of them constitute an application
\FORALL {application $i$}
\REPEAT
\STATE Call Algorithm \ref{algMultiLine}
for application $i$ with ${p}{}_{n,k}(i-1)=\max\left\{ 0,{p}{}_{n,k}(i-1)-{p}{}_{n,k}(i_{0}-1)\right\} $ and ${q}{}_{l}(i-1)=\max\left\{ 0,{q}{}_{l}(i-1)-{q}{}_{l}(i_{0}-1)\right\} $
\IF {Algorithm \ref{algMultiLine} returns FAIL}
\STATE Set $\hat{J}\leftarrow 2\hat{J}$
\STATE Set $i_{0}\leftarrow i$
\ENDIF
\UNTIL {Algorithm \ref{algMultiLine} does not return FAIL}
\STATE Map application $i$ according to $\pi_{i}$
resulting from Algorithm \ref{algMultiLine}
\ENDFOR
\ENDIF
\ENDLOOP
\end{algorithmic}
}
\end{algorithm}

\subsubsection{Complexity and Competitive Ratio}

\label{sec:online-fixed-compratio} In the following, we discuss the complexity and competitive ratio of Algorithm~\ref{algDoubling}.

Because the value of $J_{\pi^{o*}}$ is finite%
\footnote{The value of $J_{\pi^{o*}}$ is finite unless the placement cost specification does not allow any placement with finite cost.
We do not consider this case here because it means that the placement
is not realizable under the said constraints. In practice, the algorithm
can simply reject such application graphs when the mapping cost resulting
from Algorithm \ref{algMultiLine} is infinity, regardless of what
value of $\hat{J}$ has been chosen.%
}, the doubling procedure in Algorithm~\ref{algDoubling} only contains
finite steps. The remaining part of the algorithm mainly consists of
calling Algorithm~\ref{algMultiLine} which then calls Algorithm~\ref{algLineToTree} for each simple branch. Because nodes and links in each simple branch together with the set
of nodes with given placement add up to the whole
application graph, similar to Algorithm~\ref{algLineToTree}, the \emph{time-complexity of Algorithm~\ref{algDoubling} is  $O(V{}^{3}N{}^{2})$}
for each application graph arrival.

Similarly, when combining the procedures in Algorithms~\ref{algLineToTree}--\ref{algDoubling}, we can see that the  \emph{space-complexity of Algorithm~\ref{algDoubling} is  $O(VN(V+N))$}
for each application graph arrival, which is in the same order as Algorithm~\ref{algLineToTree}.

For the competitive ratio, we have the following result. 

\begin{proposition}\label{prop:competRatio} \textbf{(Competitive
Ratio)}: Algorithm \ref{algDoubling} is $4\beta=4\log_{\alpha}\left(\frac{\gamma(NK+L)}{\gamma-1}\right)$-competitive.
\end{proposition}
\begin{proof} 
If Algorithm \ref{algMultiLine} fails,
we know that $J_{\pi^{o*}}>\hat{J}$ according to Proposition
\ref{prop:Bound}. Hence, by doubling the value of $\hat{J}$ each
time  Algorithm \ref{algMultiLine} fails, we have $\hat{J}_{f}<2J_{\pi^{o*}}$,
where $\hat{J}_{f}$ is the final value of $\hat{J}$ after placing
all $M$ applications. Because we ignore all previous placements and
only consider the applications $i_{0},...,i$ for a particular value
of $\hat{J}$, it follows that 
\begin{eqnarray}
\max &  & \!\!\!\!\!\!\!\!\!\Big\{ \max_{k,n}\{{p}_{n,k}(i)-{p}_{n,k}(i_{0}-1)\}; \nonumber \\
 &  & \max_{l}\{{q}_{l}(i)-{q}_{l}(i_{0}-1)\}\Big\} \leq\beta\hat{J}\label{eq:doublingProof}
\end{eqnarray}
for the particular value of $\hat{J}$.

When we consider all the placements of $M$ applications, by summing
up (\ref{eq:doublingProof}) for all values of $\hat{J}$, we have
\begin{align*}
 \max & \!\left\{\! \max_{k,n}{p}_{n,k}(M);\max_{l}{q}_{l}(M)\!\right\} 
 \leq \left(\!1\!+\!\frac{1}{2}\!+\!\frac{1}{4}\!+\!\frac{1}{8}\!+\!\cdots\!\right)\!\beta\hat{J}_{f}\\
 & < 2\left(1+\frac{1}{2}+\frac{1}{4}+\frac{1}{8}+\cdots\right)\beta J_{\pi^{o*}}
 = 4\beta J_{\pi^{o*}}.
\end{align*}
Hence, the proposition follows.
\end{proof}

The variables $\alpha$, $\gamma$ and $K$ are constants, and $L=N-1$
because the physical graph is a tree. Hence, the competitive ratio
of Algorithm \ref{algDoubling} can also be written as $O(\log N)$.

It is also worth noting that, for each application graph, we can have
different tree physical graphs that are extracted from a general physical
graph, and the above conclusions still hold.

\subsection{When at Least One Junction Node Placement Is Not Given}

\label{sec:UnfixedDegree3}

In this subsection, we focus on cases where the placements of some
or all junction nodes are not given.
For such scenarios, we first extend our concept of branches to incorporate some unplaced junction nodes. The basic idea is
that each \emph{general branch }is the largest subset of nodes and
edges that are interconnected with each other not including any of the nodes with pre-specified placement, but (as with our previous definition of simple branches)
the subset includes the edges connected to placed nodes. A simple branch (see definition in Section \ref{sec:FixedDegree3})
is always a general branch, but a general branch may or may not be a simple branch. Examples of general
branches are shown in Fig. \ref{fig:exampleUnfixedDegree3}.

\begin{figure}
\center{\includegraphics[width=0.9\columnwidth]{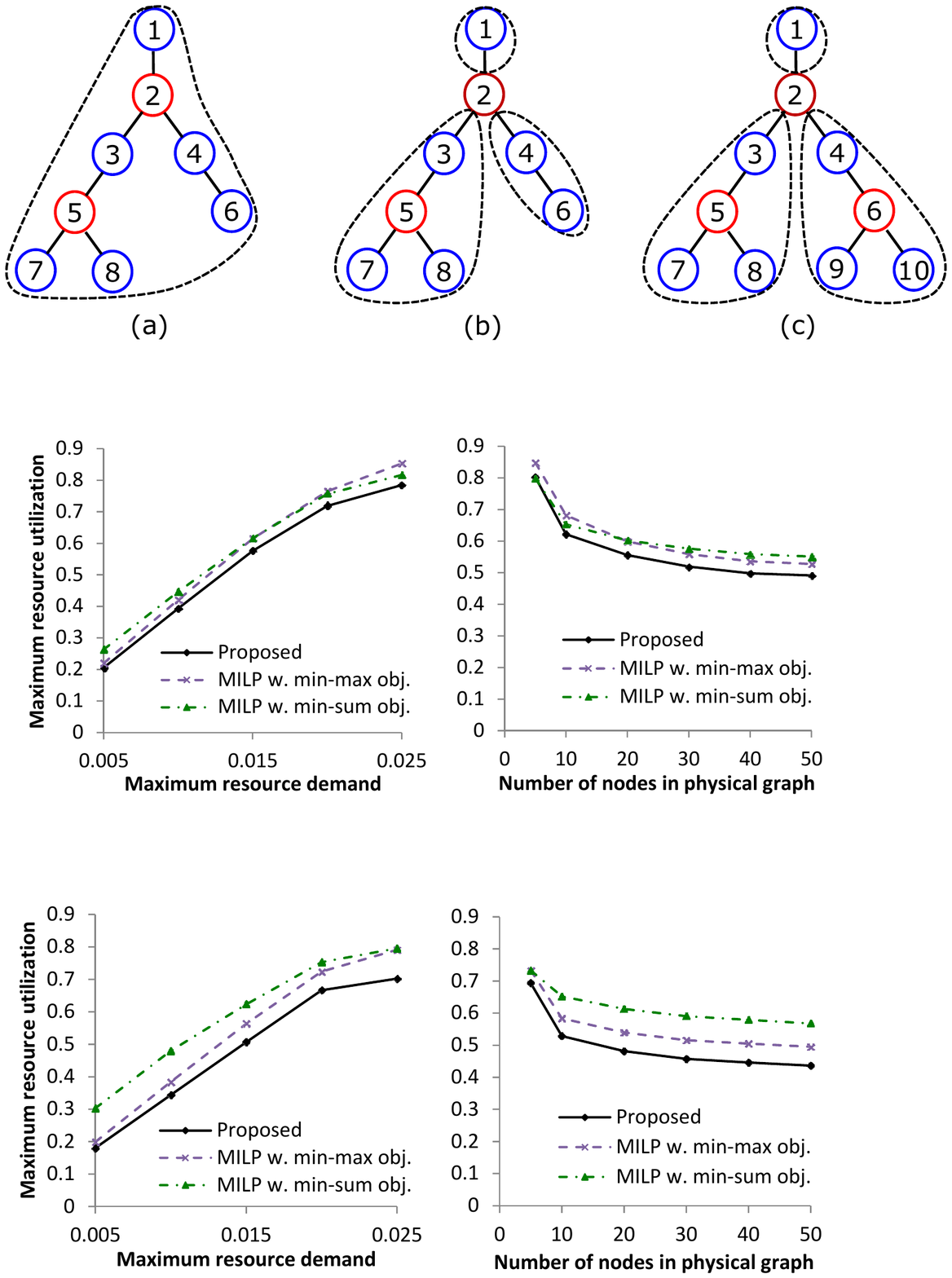}}

\caption{Example of application graphs with some unplaced junction nodes, the nodes and edges within each dashed boundary
form a general branch: (a)
nodes $2$ and $5$ are both unplaced, (b) node $2$ is placed, node $5$ is unplaced, (c) node $2$ is
placed, nodes $5$ and $6$ are unplaced.}

\label{fig:exampleUnfixedDegree3} 
\end{figure}

\subsubsection{Algorithm Design}

\label{sec:unfixed-algorithm} The main idea behind the algorithm
is to combine Algorithm \ref{algMultiLine} with the enumeration of
possible placements of unplaced junction nodes. When
there is only a constant number of such nodes on any path from the root to a leaf, the algorithm
remains polynomial in time-complexity while guaranteeing a polynomial-logarithmic (poly-log)
competitive ratio.

To illustrate the intuition, consider the example application graph shown in Fig. \ref{fig:exampleUnfixedDegree3}(a),
where nodes $2$ and $5$ are both initially unplaced.
We follow a hierarchical determination of the placement of unplaced nodes
starting with the nodes in the deepest level. For the example in Fig.
\ref{fig:exampleUnfixedDegree3}(a), we first determine the placement
of node $5$, given each possible placement of node $2$; then determine the placement of node $2$. Recall that
we use the cost function in (\ref{eq:objOfflineExtSum}) with (\ref{eq:onlineExpObj_fnk}) and (\ref{eq:onlineExpObj_fl}) to determine
the placement of each simple branch when all the junction nodes are placed. We use the same cost function (with slightly modified parameters) for the placement of nodes
$2$ and $5$. However, when determining the placement of node $5$, we regard
the general branch that includes node $5$ (which contains nodes $3$,
$5$, $7$, and $8$ and the corresponding edges as shown in Fig. \ref{fig:exampleUnfixedDegree3}(b))
as one single application, i.e. the values of ${p}{}_{n,k}(i-1)$
and ${q}{}_{l}(i-1)$ in (\ref{eq:onlineExpObj_fnk}) and (\ref{eq:onlineExpObj_fl}) correspond
to the resource utilization at nodes and links before placing this whole general branch,
and the application $i$ contains all the nodes and edges in this general
branch. Similarly, when determining the placement of node $2$, we consider
the whole application graph as a single application.

It is worth noting that in many cases we may not need to enumerate
all the possible combinations of the placement of unplaced
junction nodes. For example, in Fig. \ref{fig:exampleUnfixedDegree3}(c),
when the placement of node $2$ is given, the placement of nodes $5$ and
$6$ does not impose additional restrictions upon each
other (i.e., the placement of node
$5$ does not affect where node $6$ can be placed, for instance). Hence, the general branches that respectively include node
$5$ and node $6$ can be placed in a subsequent order using the online
algorithm.

Based on the above examples, we summarize the procedure
as Algorithm \ref{algTreeToTreeNotFixedSub}, where we solve the problem
recursively and determine the placement of one junction node that has not been placed before in each instance of the function Unplaced($v,h$).
The parameter
$v$ is initially set to the top-most unplaced junction node (node $2$ in Fig. \ref{fig:exampleUnfixedDegree3}(a)), and
$h$ is initially set to $H$ (the maximum number of unplaced junction nodes on any path from the root to a leaf in the application
graph).

\begin{algorithm}
\caption{Tree-to-tree placement when some junction nodes are not placed}
\label{algTreeToTreeNotFixedSub}
{
\begin{algorithmic}[1]
\STATE \textbf{function} Unplaced($v,h$)
\STATE Given the $i$th application that is a general branch, tree
physical graph $\mathcal{{Y}}$, $\hat{J}$, and $\beta$
\STATE Given ${p}{}_{n,k}(i-1)$ and ${q}{}_{l}(i-1)$
which is the current resource utilization on nodes and links 
\STATE Define $\mathbf{\Pi}$ to keep the currently obtained mappings, its entry $\pi|_{\pi(v)=n_{0}}$ for all $n_{0}$ represents the mapping, given that $v$ is mapped to $n_{0}$ \label{algTreeToTreeNotFixedSub:line:varDef1}
\STATE Define ${p}{}_{n,k}(i)|_{\pi(v)=n_{0}}$ and ${q}{}_{l}(i)|_{\pi(v)=n_{0}}$ for all $n_0$ as the resource utilization after placing the $i$th application, given that $v$ is mapped to $n_{0}$ \label{algTreeToTreeNotFixedSub:line:varDef2}
\STATE Initialize ${p}{}_{n,k}(i)|_{\pi(v)=n_{0}}\leftarrow {p}{}_{n,k}(i-1)$ and ${q}{}_{l}(i)|_{\pi(v)=n_{0}}\leftarrow {q}{}_{l}(i-1)$ for all $n_0$ \label{algTreeToTreeNotFixedSub:line:varDef3}
\FORALL {$n_{0}$ that $v$ can be mapped to}
\STATE Assume $v$ is placed at $n_{0}$
\FORALL {general branch that is connected
with $v$}
\IF {the general branch contains
unplaced junction nodes}
\STATE Find the top-most unplaced vertex
$v'$ within this general branch
\STATE Call Unplaced($v',h-1$)
while assuming $v$ is placed at $n_{0}$, and with ${p}{}_{n,k}(i-1)={p}{}_{n,k}(i)|_{\pi(v)=n_{0}}$
and ${q}{}_{l}(i-1)={q}{}_{l}(i)|_{\pi(v)=n_{0}}$
\ELSE 
\STATE (in which case the general
branch is a simple branch without unplaced junction nodes) \\ Run Algorithm \ref{algMultiLine}
for this branch
\ENDIF

\STATE Put mappings resulting from
Unplaced($v',h-1$) or Algorithm \ref{algMultiLine} into
$\pi|_{\pi(v)=n_{0}}$
\STATE Update ${p}{}_{n,k}(i)|_{\pi(v)=n_{0}}$
and ${q}{}_{l}(i)|_{\pi(v)=n_{0}}$ to incorporate new mappings
\ENDFOR
\ENDFOR
\STATE Find $\min_{n_{0}}\!\sum_{k,n}\!\!\left(\!\exp_\alpha\!\!\left(\!{\frac{{p}{}_{n,k}(i)|_{\pi(v)=n_{0}}}{\beta^{h}\hat{J}}}\!\right)\!\!-\!\exp_\alpha\!\!\left(\!{\frac{{p}{}_{n,k}(i-1)}{\beta^{h}\hat{J}}}\!\right)\!\right)\! +\sum_{l}\left(\exp_\alpha\left({\frac{{q}{}_{l}(i)|_{\pi(v)=n_{0}}}{\beta^{h}\hat{J}}}\right)-\exp_\alpha\left({\frac{{q}{}_{l}(i-1)}{\beta^{h}\hat{J}}}\right)\right)$, \label{algUnfixedMinLine} \\
returning the
optimal placement of $v$ as $n_{0}^{*}$.
\IF { $h=H$ and (\,$\exists n,k:{p}{}_{n,k}(i)|_{\pi(v)=n_{0}^{*}}>\beta^{1+H}\hat{J}$ or $\exists l:{q}{}_{l}(i)|_{\pi(v)=n_{0}^{*}}>\beta^{1+H}\hat{J}$\,)}
\STATE \textbf{return} FAIL
\ELSE 
\STATE \textbf{return} $\pi|_{\pi(v)=n_{0}^{*}}$
\ENDIF
\end{algorithmic}
}
\end{algorithm}

Algorithm \ref{algTreeToTreeNotFixedSub} can be embedded into Algorithm
\ref{algDoubling} to handle multiple arriving application graphs
and unknown reference cost $\hat{J}$. The only part that needs to be modified in Algorithm~\ref{algDoubling}
is that it now splits the whole application graph into general branches
(rather than simple branches without unplaced junction nodes), and it either calls Algorithm \ref{algMultiLine} or Algorithm
\ref{algTreeToTreeNotFixedSub} depending on whether there are unplaced
junction nodes in the corresponding general branch.
When there are such nodes, it calls Unplaced($v,h$) with
the aforementioned initialization parameters.


\subsubsection{Complexity and Competitive Ratio}
\label{sec:online-unfixed-competRatio}

The \emph{time-complexity of Algorithm \ref{algTreeToTreeNotFixedSub}} together
with its high-level procedure that is a modified version of Algorithm
\ref{algDoubling} is $O(V{}^{3}N{}^{2+H})$ for each application
graph arrival, as explained below. Note that $H$ is generally not the total number of unplaced nodes.

Obviously, when $H=0$, the time-complexity is the same as the
case where all junction nodes are placed beforehand.  When there is only one unplaced junction node (in which case $H=1$), Algorithm~\ref{algTreeToTreeNotFixedSub}
considers all possible placements for this vertex, which has at most
$N$ choices. Hence, its time-complexity becomes $N$ times the time-complexity
with all placed junction nodes. When there are multiple unplaced
junction nodes, we can see from Algorithm \ref{algTreeToTreeNotFixedSub}
that it only increases its recursion depth when some lower
level unplaced junction nodes exist. In other words, parallel general
branches (such as the two general branches that respectively include
node $5$ and node $6$ in Fig. \ref{fig:exampleUnfixedDegree3}(c)) do
not increase the recursion depth, because the function Unplaced($v,h$)
for these general branches is called in a sequential order. Therefore, the
time-complexity depends on the maximum recursion depth which is $H$; thus,
the overall time-complexity is $O(V{}^{3}N{}^{2+H})$.

The \emph{space-complexity of Algorithm~\ref{algTreeToTreeNotFixedSub} is  $O(VN^{1+H}(V+N))$} for each application graph arrival, because in every recursion, the results for all possible placements of $v$ are stored, and there are at most $N$ such placements for each junction node.


Regarding the competitive ratio, similar to Proposition~\ref{prop:Bound},
we can obtain the following result.

\begin{proposition}\label{prop:BoundUnfixed}If there exists an offline
mapping $\pi^{o}$ that considers all $M$ application graphs and brings cost
$J_{\pi^{o}}$, such that $J_{\pi^{o}}\leq\hat{J}$, then Algorithm
\ref{algTreeToTreeNotFixedSub} never fails, i.e., ${p}_{n,k}(M)$
and ${q}_{l}(M)$ resulting from Algorithm \ref{algTreeToTreeNotFixedSub}
never exceeds $\beta^{1+H}\hat{J}$.
\end{proposition}
\begin{proof}
When $H=0$, the claim is the
same as Proposition \ref{prop:Bound}.
When $H=1$, there is at most one unplaced junction node in each general branch. Because Algorithm~\ref{algTreeToTreeNotFixedSub}
operates on each general branch, we can regard that we have only
one unplaced junction node when running Algorithm~\ref{algTreeToTreeNotFixedSub}.
In this case, there is no recursive calling of Unplaced($v,h$).
Recall that $v$ is the top-most unplaced junction node. The function Unplaced($v,h$) first fixes the placement of  $v$ to a particular physical node $n_{0}$, and finds the placement of
the remaining nodes excluding $v$. It then finds the placement of $v$.

From Proposition \ref{prop:Bound}, we know that when we fix the placement
of $v$, the cost resulting from the algorithm never exceeds
$\beta\hat{J}$ if there exists a mapping $\pi^{o}|_{\pi(v)=n_0}$ (under the constraint
that $v$ is placed at $n_{0}$) that brings cost 
$J_{\pi^{o}|_{\pi(v)=n_0}}\leq\hat{J}$. 

To find the placement of $v$, Algorithm~\ref{algTreeToTreeNotFixedSub}
finds the minimum cost placement from the set of placements that have been obtained
when the placement of $v$ is given. 
Reapplying Proposition~\ref{prop:Bound} for the placement of $v$, by substituting $\hat{J}$
with $\beta\hat{J}$, we know that the cost from the algorithm
never exceeds $\beta^{2}\hat{J}$, provided that there exists a mapping, which
is within the set of mappings produced by the algorithm with given
$v$ placements\footnote{Note that, as shown in Line \ref{algUnfixedMinLine} of Algorithm \ref{algTreeToTreeNotFixedSub},  to determine the placement of $v$, we only take the minimum  cost (expressed as the difference of exponential functions) with respect to those mappings that were obtained with given placement of $v$. It follows that the minimization is only taken among a subset of all the possible mappings. This restricts the reference mapping to be within the set of mappings that the minimization operator operates on. Because, only in this way, the inequality (\ref{app:ProofGraphPlacementApproxBound:eq:potentialChange3}) in the proof of Proposition \ref{prop:Bound} can be satisfied. On the contrary, Algorithm \ref{algMultiLine} considers all possible mappings that a particular simple branch can be mapped to, by calling Algorithm \ref{algLineToTree} as its subroutine.}, that has a cost not exceeding $\beta\hat{J}$. Such
a mapping exists and can be produced by the algorithm if there exists
an offline mapping $\pi^{o}$ (thus a mapping $\pi^{o}|_{\pi(v)=n_0}$ for a particular placement of $v$) that brings cost $J_{\pi^{o}}$ with
$J_{\pi^{o}}\leq\hat{J}$. Hence, the claim follows for $H=1$.

When $H>1$, because we decrease the value of $h$ by one every time
we recursively call Unplaced($v,h$), the same propagation
principle of the bound applies as for the case with $H=1$. Hence,
the claim follows.
\end{proof}

Using the same reasoning as for Proposition \ref{prop:competRatio},
it follows that \emph{Algorithm \ref{algTreeToTreeNotFixedSub} in combination
with the extended version of Algorithm \ref{algDoubling} is $4\beta^{1+H}=4\log_{\alpha}^{1+H}\left(\frac{\gamma(NK+L)}{\gamma-1}\right)$-competitive, thus its competitive ratio is $O(\log^{1+H}N)$ }.

\section{Numerical Evaluation}
\label{sec:SimulationResults}
We compare the proposed algorithm against two heuristic approaches via simulation.
The first approach is one that greedily minimizes the maximum resource utilization (according to (\ref{eq:objOnline})) for the placement of \emph{every} newly arrived application graph. The second approach is the Vineyard algorithm proposed in \cite{refViNEYard}, where load balancing is also considered as a main goal in application placement. 

\begin{figure*}
\center{\includegraphics[width=1.35\columnwidth]{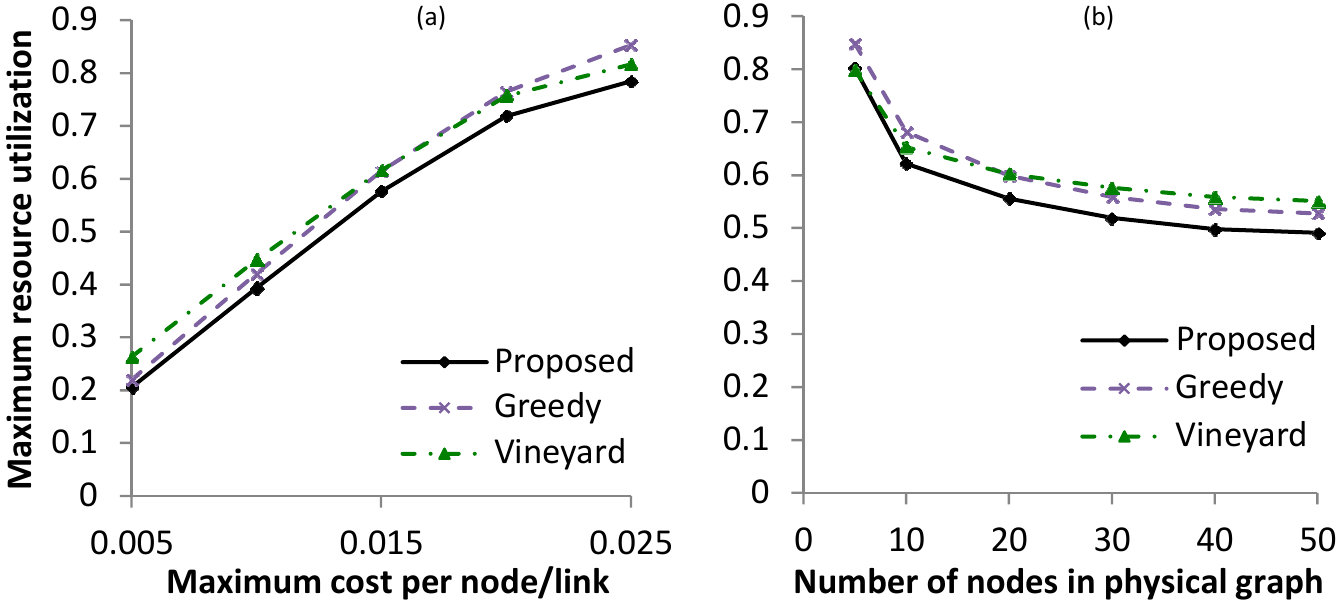}}
\caption{Maximum resource utilization when junction node placements are pre-specified.}
\label{fig:simFixed} 
\end{figure*}

\begin{figure*}
\center{\includegraphics[width=1.35\columnwidth]{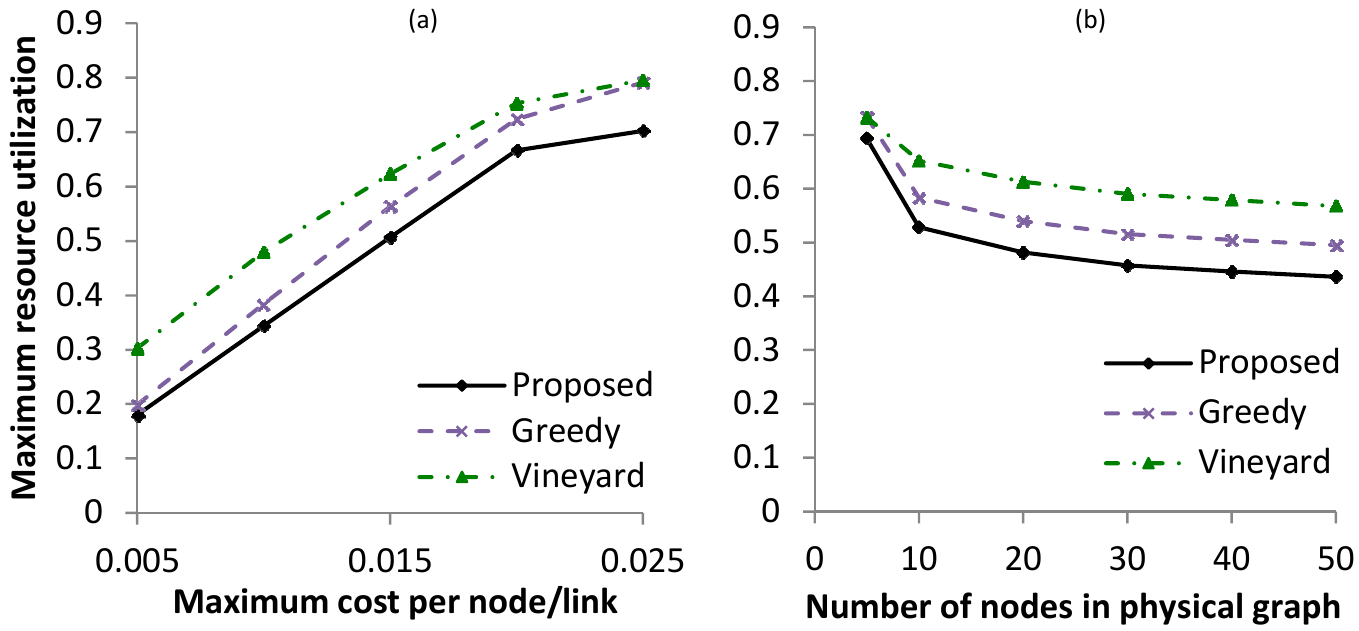}}
\caption{Maximum resource utilization when junction node placements are not pre-specified.}
\label{fig:simUnfixed} 
\end{figure*}

Both the greedy and Vineyard algorithms require an optimization problem to be solved as a subroutine, for the placement of every newly arrived application. This optimization problem can be expressed as a mixed-integer linear program (MILP). MILPs are generally \emph{not} solvable in polynomial-time, thus an LP-relaxation and rounding procedure is used in \cite{refViNEYard}.
In this paper, to capture the best generality and eliminate inaccuracies caused by heuristic rounding mechanisms (because there are multiple ways of rounding that one could use), we solve the MILP subroutines directly using CPLEX \cite{IBMCPLEX}.
This gives an exact solution to the subroutine, thus the greedy and Vineyard algorithms in the simulation may perform better than they would in reality, and we are conservative in showing the effectiveness of the proposed algorithm. 

Note that these MILP solutions do \emph{not} represent the optimal offline solution, because an optimal offline solution needs to consider all application graphs at the same time, whereas the methods that we use for comparison only solve the MILP subroutine for each newly arrived application. Obtaining the optimal offline solution requires excessive computational time such that the simulation infeasible, hence we do not consider it here. 
We also do not compare against the theoretical approach in  \cite{PODC2011} via simulation, because that approach is non-straightforward to implement. However, we have outlined the benefits of our approach against \cite{PODC2011} in Section \ref{sub:contributionsInIntro} and some further discussion will be given in Section \ref{sec:discussion}.

To take into account possible negative impacts of the cycle-free restriction in the proposed algorithm, we do \emph{not} impose the cycle-free constraint in the baseline greedy and Vineyard algorithms. However, for a fair comparison, we do require in the baseline approaches that when the placements of junction nodes are given, the children of this junction node can only be placed onto the physical node on which the junction node has been placed, or onto the children of this physical node. 


Because MEC is a very new concept which has not been practically deployed in a reasonably large scale, we currently do not have real topologies available to us for evaluation. Therefore, similar to existing work such as \cite{refViNEYard}, we consider synthetic tree application and physical graphs.
Such graphs mimic realistic MEC setups where MEC servers and applications locate at multiple network equipments in different hierarchical levels, see \cite{ETSIWhitepaper} and the example in Section \ref{sec:intro} for instance.
The number of application nodes for each application is randomly chosen from the interval $[3,10]$, and the number of physical nodes ranges from $2$ to $50$. This simulation setting is similar to that in \cite{refViNEYard}. We use a sequential approach to assign connections between nodes. Namely, we first label the nodes with indices. Then, we start from the lowest index, and connect each node $m$ to those nodes that have indices $1,2,...,m-1$. Node $m$ connects to node $m-1$ with probability $0.7$, and connects to nodes $1,2,...,m-2$ each with probability $0.3/(m-2)$.
We restrict the application root node to be placed onto the physical root node, considering that some portion of processing has to be performed on the core cloud possibly due to the constraint of database location (see Fig. \ref{chap:intro:fig:scenario}). We consider $100$ application arrivals and simulate with $100$ different random seeds to obtain the overall performance. The placement cost of a single node or link is uniformly distributed between $0$ and a maximum cost. For the root application node, the cost is divided by a factor of $10$. We set the design parameter $\gamma=2$. 

\begin{figure*}
\center{\includegraphics[width=1.3\columnwidth]{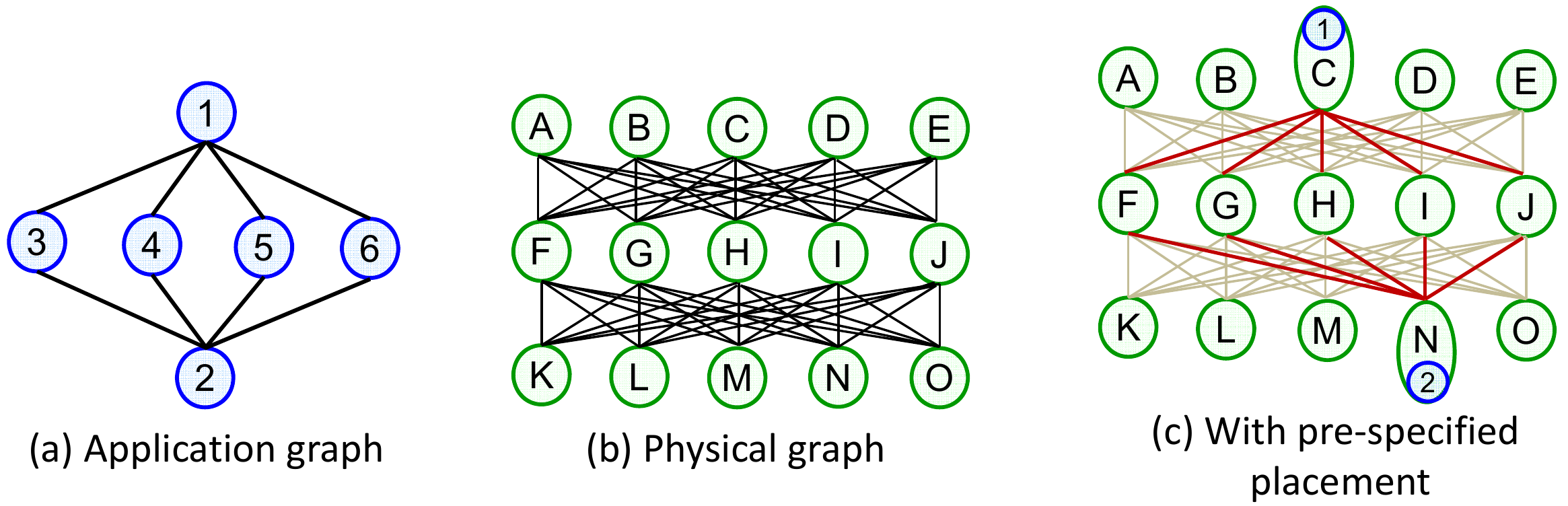}}
\caption{Example where application and physical graphs are not trees: (a) application graph, (b) physical graph, (c) restricted physical graph with pre-specified placement of application nodes $1$ and $2$.}
\label{fig:nonTreeExample} 
\end{figure*}

Figures~\ref{fig:simFixed} and \ref{fig:simUnfixed} show the maximum resource utilization, i.e., the value of (\ref{eq:objOnline}), averaged over results from different random seeds\footnote{We only consider those random seeds which produce a maximum resource utilization that is smaller than one, because otherwise, the physical network is considered as overloaded after hosting 100 applications. We also observed in the simulations that the number of accepted applications is similar when using different methods. The relative degradation in the number of accepted applications of the proposed method compared with other methods never exceeds 2\% in the simulations.}, respectively with and without pre-specified placement of junction nodes.
In Figs.~\ref{fig:simFixed}(a) and \ref{fig:simUnfixed}(a), the number of physical nodes is randomly chosen from the interval $[2,50]$; and in Figs.~\ref{fig:simFixed}(b) and \ref{fig:simUnfixed}(b), the maximum cost per application node/link is set to $0.015$.
It is evident that the proposed method outperforms those methods in comparison. The resource utilization tends to converge when the number of physical nodes is large because of the fixed root placement. As mentioned earlier, practical versions of greedy and Vineyard algorithms that have LP-relaxation and rounding may perform worse than what our current results show.

We now explain why the proposed method outperforms other methods. We first note that the uniqueness in the proposed algorithm is that it uses a non-linear objective function for placing each new application, whereas the baseline methods and most other existing approaches use linear objective functions. 
The exponential-difference cost (\ref{eq:objOfflineExtSum}) with (\ref{eq:onlineExpObj_fnk}) and (\ref{eq:onlineExpObj_fl}) used in the proposed algorithm for the placement of each newly arrived application graph aims at both load balancing and reducing sum resource utilization. It leaves more space for applications that arrive in the future. Therefore, it outperforms the greedy approach which does not take future arrivals into account.
The Vineyard approach does not strongly enforce load balancing unless operating close to the resource saturation point, due to the characteristics of its objective function used in each subroutine of application arrival.

When comparing Fig.~\ref{fig:simFixed} to  Fig.~\ref{fig:simUnfixed}, we can find that the performance gaps between the proposed method and other methods are larger when the junction nodes are not placed beforehand. This is mainly because the judgment of whether Algorithm \ref{algTreeToTreeNotFixedSub} has failed is based on the factor $\beta^{1+H}$, and for Algorithm \ref{algMultiLine} it is based on $\beta$. It follows that Algorithm \ref{algTreeToTreeNotFixedSub} is less likely to fail when $H>0$. In this case, the value of $\hat{J}$ is generally set to a smaller value by the doubling procedure in Algorithm \ref{algDoubling}. A smaller value of $\hat{J}$ also results in a larger change in the exponential-difference cost when the amount of existing load changes\footnote{This is except for the top-level instance of Unplaced($v,h$) due to the division by $\beta^h$ in Line \ref{algUnfixedMinLine} of Algorithm \ref{algTreeToTreeNotFixedSub}.}. This brings a better load balancing on average (but not for the worst-case, the worst-case result is still bounded by the bounds derived earlier in this paper).

\section{Discussion}
\label{sec:discussion}
\textbf{Is the Tree Assumption Needed?} For ease of presentation and considering the practical relevance to MEC applications, we have focused on tree-to-tree placements in this paper. However, the tree assumption is \emph{not} absolutely necessary for our algorithms to be applicable. 
For example, consider the placement problem shown in Fig. \ref{fig:nonTreeExample}, where the application graph contains two junction nodes\footnote{For non-tree graphs, a junction node can be defined as those nodes that are not part of a simple branch.} (nodes $1$ and $2$) and multiple simple branches (respectively including nodes $3$, $4$, $5$, and $6$) between these two junction nodes. Such an application graph is common in applications where processing can be parallelized at some stage. The physical graph shown in Fig. \ref{fig:nonTreeExample}(b) still has a hierarchy, but we now have connections between all pairs of nodes at two adjacent levels. Obviously, neither the application nor the physical graph in this problem has a tree structure.

Let us assume that junction node $1$ has to be placed at the top level of the physical graph (containing nodes A, B, C, D, E), junction node $2$ has to be placed at the bottom level of the physical graph (containing nodes K, L, M, N, O), and application nodes $3,4,5,6$ have to be placed at the middle level of the physical graph (containing nodes F, G, H, I, J). One possible junction node placement under this restriction is shown in Fig. \ref{fig:nonTreeExample}(c). With this pre-specified junction node placement, the mapping of each application node in $\{3,4,5,6\}$ can be found by the simple branch placement algorithm (Algorithm \ref{algDoubling} which embeds Algorithm \ref{algMultiLine}) introduced earlier, because it only needs to map each application node in $\{3,4,5,6\}$ onto each physical node in $\{$F, G, H, I, J$\}$, and find the particular assignment that minimizes (\ref{eq:objOfflineExtSum}) with (\ref{eq:onlineExpObj_fnk}) and (\ref{eq:onlineExpObj_fl}).
Therefore, in this example, when the junction node placements are pre-specified, the proposed algorithm can find the placement of other application nodes with $O(V{}^{3}N{}^{2})$ time-complexity, which is the complexity of Algorithm~\ref{algDoubling} as discussed in Section~\ref{sec:online-fixed-compratio}.
When the junction node placements are not pre-specified, the proposed algorithm can find the placement of the whole application graph with $O(V{}^{3}N{}^{4})$ time-complexity, because here $H=2$ (recall that the complexity result was derived in Section \ref{sec:online-unfixed-competRatio}).

We envision that this example can be generalized to a class of application and physical graphs where there exist a limited number of junction nodes that are not placed beforehand. The algorithms proposed in this paper should still be applicable to such cases, as long as we can find a limited number of cycle-free paths between two junction nodes when they are placed on the physical graph. We leave a detailed discussion on this aspect as future work.

\textbf{Practical Implications:} 
Besides the proposed algorithms themselves, the results of this paper also reveal the following insights that may guide future implementation:
\begin{enumerate}
\item The placement is easier when the junction nodes are placed beforehand. This is obvious when comparing the time-complexities and competitive ratios for cases with and without unplaced junction nodes.
\item There is a trade-off between instantaneously satisfying the objective function and leaving more available resources for future applications. Leaving more available resources may cause the system to operate in a sub-optimal state for the short-term, but future applications may benefit from it. This trade-off can be controlled by defining an alternative objective function which is different from (but related to) the overall objective that the system tries to achieve (see Section \ref{sub:unplacedAlgDesign}).
\end{enumerate}

\textbf{Performance Bound Comparison:} As mentioned in Section \ref{sec:intro}, \cite{PODC2011} is the only work which we know that has studied the competitive ratio of online application placement considering both node and link optimization. Our approach has several benefits compared to \cite{PODC2011} as discussed in Section \ref{sub:contributionsInIntro}. Besides those benefits, we would like to note that the proposed algorithm outperforms \cite{PODC2011} in time-complexity, space-complexity, and competitive ratio when the placements of all junction nodes (if any) are pre-specified. The performance bounds of these two approaches can be found in Sections \ref{sub:relatedWork} and  \ref{sec:mainResults}, respectively. 
Note that a linear application graph does not have any junction node, thus it falls into the above category. Linear application graphs are the case for a typical class of MEC applications (see the example in Section \ref{subsec:MotivatingExample}) as well as for related problems such as service chain embedding \cite{serviceChainInfocom2016,rost2016service,Lukovszki2015}.
When some junction node placements are not pre-specified, our approach provides a performance bound comparable to that in \cite{PODC2011}, because $H\leq D$. Moreover, \cite{PODC2011} does not provide exact optimal solutions for the placement of a single linear application graph; it also does not have simulations to show the average performance of the algorithm.

\textbf{Tightness of Competitive Ratio:} 
By comparing the competitive ratio result of our approach to that in \cite{PODC2011}, we see that both approaches provide poly-log competitive ratios for the general case. It is however unclear whether this is the best performance bound one can achieve for the application placement problem. This is an interesting but difficult aspect worth studying in the future.

\section{Conclusions}
\label{sec:conclusion}

In this paper, the placement of an incoming stream of application graphs onto a physical graph has been studied under the MEC context. 
We have first proposed an exact optimal algorithm for placing one linear application graph onto a tree physical graph which works for a variety of objective functions.
Then, with the goal of minimizing the maximum resource utilization at physical nodes and links, we have proposed online approximation algorithms for placing tree application graphs onto tree physical graphs. When the maximum number of unplaced junction nodes on any path from the root to a leaf (in the application graph) is a constant, the proposed algorithm has polynomial time and space complexity and provides poly-log worst-case optimality bound (i.e., competitive ratio). Besides the theoretical evaluation of worst-case performance, we have also shown the average performance via simulation. A combination of these results implies that the proposed method performs reasonably well on average and it is also robust in extreme cases. 

The results in this paper can be regarded as an initial step towards a more comprehensive study in this direction. Many constraints in the problem formulation are for ease of presentation, and can be readily relaxed for a more general problem.
For example, as discussed in Section~\ref{sec:discussion}, the tree-topology restriction is not absolutely essential for the applicability of our proposed algorithms. The algorithms also work for a class of more general graphs as long as the cycle-free constraint is satisfied. 
While we have not considered applications leaving at some time after their arrival, our algorithm can be extended to incorporate such cases, for example using the idea in \cite{AzarTemporaryTasks}. The algorithm for cases with unplaced junction nodes is essentially considering the scenario where there exists some low-level placement (for each of the branches) followed by some high level placement (for the junction nodes). Such ideas may also be useful in developing practical distributed algorithms with provable performance guarantees.


\appendices












\section{Approximation Ratio for Cycle-free Mapping}

\label{app:ApproxRatioCycleFree}

We focus on how well the cycle-free restriction approximates
the more general case which allows cycles, for the placement of a
single linear application graph. We first show that with the objective
of load balancing (defined in (\ref{eq:objOnline}) in Section \ref{subSec:overallObjectiveinProbFormulation}), the problem of placing a single linear application
graph onto a linear physical graph when allowing cycles is NP-hard, and then discuss the
approximation ratio of the cycle-free restriction.

\begin{proposition}The line-to-line placement problem for the objective function defined in (\ref{eq:objOnline}) \emph{while allowing
cycles} is NP-hard. \end{proposition}

\begin{proof}The proof is similar with the proof of Proposition~\ref{propTreeToTreeNPHard} in Section \ref{sec:hardnessTreeToTree},
namely the problem can be reduced from the minimum makespan scheduling
on unrelated parallel machines (MMSUPM) problem. Consider the special
case where the edge demand is zero, then the problem is the same with
the MMSUPM problem, which deals with placing $V$ jobs onto $N$ machines
without restriction on their ordering, with the goal of minimizing
the maximum load on each machine. \end{proof}

To discuss the approximation ratio of the cycle-free assignment, we
separately consider edge costs and node costs. The worst case ratio
is then the maximum among these two ratios, because we have $\max\left\{ r_{1}x_{1},r_{2}x_{2}\right\} \leq\max\left\{ r_{1},r_{2}\right\} \cdot \max\left\{ x_{1},x_{2}\right\} $,
for arbitrary $r_{1},r_{2},x_{1},x_{2}\geq0$. The variables $x_{1}$ and
$x_{2}$ can respectively denote the true optimal maximum costs at
nodes and links, and the variables $r_{1}$ and $r_{2}$ can be their
corresponding approximation ratios. Then, $\max\left\{ x_{1},x_{2}\right\} $
is the true optimal maximum cost when considering nodes and links
together, and $\max\left\{ r_{1},r_{2}\right\} $ is their joint approximation
ratio. The joint approximation ratio $\max\left\{ r_{1},r_{2}\right\} $
is tight (i.e., there exists a problem instance where the actual optimality gap is arbitrarily close the approximation ratio, recall that the approximation ratio is defined in an upper bound sense) when $r_{1}$ and $r_{2}$ are tight, because we can construct
worst-case examples, one with zero node demand and another with zero
link demand, and there must exist one worst-case example which has
approximation ratio $\max\left\{ r_{1},r_{2}\right\} $.

In the following discussion, we assume that the application and physical nodes are indexed in the way described in Section \ref{subSec:lineToLineProbFormulation}.
The following proposition shows that cycle-free placement is always optimal
when only the edge cost is considered.

\begin{proposition}Cycle-free placement on tree physical graphs always
has lower or equal maximum edge cost compared with placement that allows
cycles. \end{proposition}

\begin{proof}Suppose a placement that contains cycles produces a lower
maximum edge cost than any cycle-free placement, then there exists
$v$ and $v_{1}$ $(v_{1}>v+1)$ both placed on a particular node
$n$, while nodes $v+1,...,v_{1}-1$ are placed on some nodes among
$n+1,...,N$. In this case, placing nodes $v+1,...,v_{1}-1$ all onto
node $n$ never increases the maximum edge cost, which shows a contradiction.
\end{proof}

For the node cost, we first consider the case where the physical graph
is a single line. We note that in this case the cycle-free placement
essentially becomes an ``ordered matching'', which
matches $V$ items into $N$ bins, where the first bin may contain
items $1,...,v_{1}$, the second bin may contain items $v_{1}+1,...,v_{2}$,
and so on. We can also view the problem as partitioning the ordered
set $\mathcal{{V}}$ into $N$ subsets, and each subset contains consecutive
elements from $\mathcal{{V}}$.

\begin{proposition} \label{prop:cycleFreeApproxRatio2} When each application node has the same cost when placing it on any physical node, then the cycle-free line-to-line
placement has a \emph{tight} approximation ratio of 2. \end{proposition}

\begin{proof} Suppose we have $V$ items that can be packed into
$N$ bins by a true optimal algorithm (which does not impose ordering
on items), and the optimal cost at each bin is $\mathrm{OPT}$.

To show that the worst case cost ratio resulting from the ordering
cannot be larger than 2, we consider a bin packing where the size of
each bin is $\mathrm{OPT}$. (Note that the bin packing problem
focuses on minimizing the number of bins with given bin size, which
is slightly different from our problem.) Because an optimal solution
can pack our $V$ items into $N$ bins with maximum cost $\mathrm{OPT}$,
when we are given that the size of each bin is $\mathrm{OPT}$, we
can also pack all the $V$ items into $N$ bins. Hence, the optimal
solution to the related bin packing problem is $N$. When we have
an ordering, we can do the bin packing by the first-fit algorithm
which preserves our ordering. The result of the first-fit algorithm
has been proven to be at most $2N$ bins \cite{bookApproxAlg}.

Now we can combine two neighboring bins into one bin. Because we have
at most $2N$ bins from the first-fit algorithm, we will have at most
$N$ bins after combination. Also because each bin has size $\mathrm{OPT}$
in the bin packing problem, the cost after combination will be at
most $2\cdot\mathrm{OPT}$ for each bin.
This shows that the worst-case cost for ordered items is at most $2\cdot\mathrm{OPT}$.

To show that the approximation ratio of 2 is tight, we consider the
following problem instance as a tight example. Suppose $V=2N$. Among the $2N$ items, $N$
of them have cost of $(1-\epsilon)\mathrm{\cdot OPT}$, where $\epsilon>\frac{1}{1+N}$,
the remaining $N$ have a cost of $\epsilon\cdot\mathrm{OPT}$. Obviously,
an optimal allocation will put one $(1-\epsilon)\mathrm{\cdot OPT}$
item and one $\epsilon\mathrm{\cdot OPT}$ item into one bin, and
the resulting maximum cost at each bin is $\mathrm{OPT}$.

A bad ordering could have all $(1-\epsilon)\mathrm{\cdot OPT}$ items
coming first, and all $\epsilon\mathrm{\cdot OPT}$ items coming afterwards.
In this case, if the ordered placement would like the maximum cost to be smaller than
$(2-2\epsilon)\mathrm{\cdot OPT}$, it would be impossible to fit
all the items into $N$ bins, because all the $(1-\epsilon)\mathrm{\cdot OPT}$
items will already occupy $N$ bins, as it is impossible to put more
than one $(1-\epsilon)\mathrm{\cdot OPT}$ item into each bin if the
cost is smaller than $(2-2\epsilon)\mathrm{\cdot OPT}$. Because $N\epsilon\mathrm{\cdot OPT}>\left(\frac{1}{\epsilon}-1\right)\epsilon\mathrm{\cdot OPT}=(1-\epsilon)\mathrm{\cdot OPT}$,
it is also impossible to put all $\epsilon\mathrm{\cdot OPT}$
into the last bin on top of the existing $(1-\epsilon)\mathrm{\cdot OPT}$
item. This means an ordered placement of these $V$ items into $N$ bins has a cost that is at least $(2-2\epsilon)\mathrm{\cdot OPT}$

Considering arbitrarily large $N$ and thus arbitrarily small $\epsilon$,
we can conclude that the approximation ratio of $2\mathrm{\cdot OPT}$ is tight. \end{proof}

\begin{corollary} When the physical graph is a tree and the maximum to minimum cost ratio for
placing application node $v$ on any physical node is ${d}_{\%,v}$, then the cycle-free line-to-line placement has an approximation ratio of $2V\cdot\max_{v}{d}_{\%,v}=O(V)$.
\end{corollary}

\begin{proof} This follows from the fact that $\mathrm{OPT}$ may
choose the minimum cost for each $v$ while the ordered assignment
may have to choose the maximum cost for some $v$, and also, in the worst
case, the cycle-free placement may place all application nodes onto
one physical node. The factor $2$ follows from Proposition \ref{prop:cycleFreeApproxRatio2}. \end{proof}

It is not straightforward to find out whether the bound in the above corollary is
tight or not, thus we do not discuss it here.

We conclude that the cycle-free placement always brings optimal link
cost, which is advantageous.
The approximation ratio of node costs can be $O(V)$ in some extreme
cases. However, the cycle-free restriction is still reasonable in
many practical scenarios. Basically, in these scenarios,
one cannot split the whole workload onto all the available servers
without considering the total link resource consumption. The analysis
here is also aimed to provide some further insights that
helps to justify in what practical scenarios the proposed work is
applicable, while further study is worthwhile for some other scenarios.

\section{Proof of Proposition \ref{prop:Bound}}
\label{app:ProofGraphPlacementApproxBound}

The proof presented here borrows ideas from \cite{OnlineRoutingWithExpCost},
but is applied here to the generalized case of graph mappings and
arbitrary reference offline costs $J_{\pi^{o}}$. 
For a given $\hat{J}$, we
define $\tilde{p}{}_{n,k}(i)={p}_{n,k}(i)/\hat{J}$, $\tilde{d}{}_{v\rightarrow n,k}(i)={d}_{v\rightarrow n,k}(i)/\hat{J}$,
$\tilde{q}{}_{l}(i)={q}_{l}(i)/\hat{J}$, and $\tilde{b}{}_{e\rightarrow l}(i)={b}_{e\rightarrow l}(i)/\hat{J}$.
To simplify the proof structure,
we first introduce some notations so that the link and node costs
can be considered in an identical framework, because it is not necessary
to distinguish them in the proof of this proposition. We refer to each type
of resources as an \emph{element}, i.e., the type $k$ resource at node
$n$ is an element, the resource at link $l$ is also an element.
Then, we define the aggregated cost up to application $i$ for element
$r$ as $\tilde{z}_{r}(i)$. The value of $\tilde{z}_{r}(i)$ can
be either $\tilde{p}{}_{n,k}(i)$ or $\tilde{q}{}_{l}(i)$ depending
on the resource type under consideration. Similarly, we define
$\tilde{w}{}_{r|\pi}(i)$ as the incremental cost that application $i$
brings to element $r$ under the mapping $\pi$. The value of $\tilde{w}{}_{r|\pi}(i)$
can be either $\sum_{\forall v:\pi(v)=n}\tilde{d}{}_{v\rightarrow n,k}(i)$
or $\sum_{\forall e=(v_{1},v_{2}):\left(\pi(v_{1}),\pi(v_{2})\right)\ni l}\tilde{b}{}_{e\rightarrow l}(i)$.
Note that both $\tilde{z}_{r}(i)$ and $\tilde{w}{}_{r|\pi}(i)$ are normalized
by the reference cost $\hat{J}$.

Using the above notations, the objective function in (\ref{eq:objOfflineExtSum}) with (\ref{eq:onlineExpObj_fnk}) and (\ref{eq:onlineExpObj_fl})
becomes 
\begin{equation}
\min_{\pi_{i}}\sum_{r}\left(\alpha^{\tilde{z}_{r}(i-1)+\tilde{w}{}_{r|\pi_{i}}(i)}-\alpha^{\tilde{z}_{r}(i-1)}\right).\label{app:ProofGraphPlacementApproxBound:eq:onlineExpObjWithElements}
\end{equation}
Note that due to the notational equivalence, (\ref{app:ProofGraphPlacementApproxBound:eq:onlineExpObjWithElements})
is the same as (\ref{eq:objOfflineExtSum}) with (\ref{eq:onlineExpObj_fnk}) and (\ref{eq:onlineExpObj_fl}).

Recall that $\pi^{o}$ denotes the reference offline mapping result,
let $\pi_{i}^{o}$ denote the offline mapping result for nodes that correspond to the $i$th application, and $\tilde{z}_{r}^{o}(i)$
denote the corresponding aggregated cost until application $i$. Define
the following potential function: 
\begin{equation}
\Phi(i)=\sum_{r}\alpha^{\tilde{z}_{r}(i)}\left(\gamma-\tilde{z}_{r}^{o}(i)\right),
\end{equation}
which helps us prove the proposition. Note that variables without superscript
``o'' correspond to the values resulting from Algorithm \ref{algMultiLine}
that optimizes the objective function (\ref{app:ProofGraphPlacementApproxBound:eq:onlineExpObjWithElements}) for each application independently.

The change in $\Phi(i)$ after new application arrival is
\begin{align}
 &\Phi(i)-\Phi(i-1)\nonumber \\
 & =  \sum_{r:\exists\pi_{i}(\cdot)=r}\left(\alpha^{\tilde{z}_{r}(i)}-\alpha^{\tilde{z}_{r}(i-1)}\right)\left(\gamma-\tilde{z}_{r}^{o}(i-1)\right) \nonumber \\
 & \quad -\sum_{r:\exists\pi_{i}^{o}(\cdot)=r}\alpha^{\tilde{z}_{r}(i)}\tilde{w}{}_{r|\pi_{i}^{o}}(i)\label{app:ProofGraphPlacementApproxBound:eq:potentialChange1}\\
 & \leq  \sum_{r:\exists\pi_{i}(\cdot)=r}\gamma\left(\alpha^{\tilde{z}_{r}(i-1)+\tilde{w}{}_{r|\pi_{i}}(i)}-\alpha^{\tilde{z}_{r}(i-1)}\right) \nonumber \\
 & \quad -\sum_{r:\exists\pi_{i}^{o}(\cdot)=r}\alpha^{\tilde{z}_{r}(i-1)}\tilde{w}{}_{r|\pi_{i}^{o}}(i)\label{app:ProofGraphPlacementApproxBound:eq:potentialChange2}\\
 & \leq  \sum_{r:\exists\pi_{i}^{o}(\cdot)=r}\gamma\left(\alpha^{\tilde{z}_{r}(i-1)+\tilde{w}{}_{r|\pi_{i}^{o}}(i)}-\alpha^{\tilde{z}_{r}(i-1)}\right) \nonumber \\
 & \quad -\alpha^{\tilde{z}_{r}(i-1)}\tilde{w}{}_{r|\pi_{i}^{o}}(i)\label{app:ProofGraphPlacementApproxBound:eq:potentialChange3}\\
 & =  \sum_{r:\exists\pi_{i}^{o}(\cdot)=r}\alpha^{\tilde{z}_{r}(i-1)}\left\{ \gamma\left(\alpha^{\tilde{w}{}_{r|\pi_{i}^{o}}(i)}-1\right)-\tilde{w}{}_{r|\pi_{i}^{o}}(i)\right\} , 
\label{app:ProofGraphPlacementApproxBound:eq:potentialChange4}
\end{align}
where the notation $\pi_{i}(\cdot)=r$ or $\pi_{i}^{o}(\cdot)=r$
means that application $i$ has occupied some resource from
element $r$ when respectively using the mapping from Algorithm \ref{algMultiLine}
or the reference offline mapping.

We explain the relationships in (\ref{app:ProofGraphPlacementApproxBound:eq:potentialChange1})--(\ref{app:ProofGraphPlacementApproxBound:eq:potentialChange4})
in the following. Equality (\ref{app:ProofGraphPlacementApproxBound:eq:potentialChange1}) follows from
\begin{align*}
 &  \Phi(i)-\Phi(i-1)\\
 & = \sum_{r}\alpha^{\tilde{z}_{r}(i)}\left(\gamma-\left(\tilde{z}_{r}^{o}(i-1)+\tilde{w}{}_{r|\pi_{i}^{o}}(i)\right)\right) \\
 & \quad -\sum_{r}\alpha^{\tilde{z}_{r}(i-1)}\left(\gamma-\tilde{z}_{r}^{o}(i-1)\right)\\
 & =  \sum_{r}\left(\alpha^{\tilde{z}_{r}(i)}-\alpha^{\tilde{z}_{r}(i-1)}\right)\left(\gamma-\tilde{z}_{r}^{o}(i-1)\right) \\
 & \quad -\sum_{r}\alpha^{\tilde{z}_{r}(i)}\tilde{w}{}_{r|\pi_{i}^{o}}(i)\\
 & =  \sum_{r:\exists\pi_{i}(\cdot)=r}\left(\alpha^{\tilde{z}_{r}(i)}-\alpha^{\tilde{z}_{r}(i-1)}\right)\left(\gamma-\tilde{z}_{r}^{o}(i-1)\right)  \\
 & \quad -\sum_{r:\exists\pi_{i}^{o}(\cdot)=r}\alpha^{\tilde{z}_{r}(i)}\tilde{w}{}_{r|\pi_{i}^{o}}(i),
\end{align*}
where the last equality follows from the fact that $\alpha^{\tilde{z}_{r}(i)}-\alpha^{\tilde{z}_{r}(i-1)}=0$
for all $r$ that $\forall\pi_{i}(\cdot)\neq r$, and $\tilde{w}{}_{r|\pi_{i}^{o}}(i)=0$
for all $r$ that $\forall\pi_{i}^{o}(\cdot)\neq r$. Inequality
(\ref{app:ProofGraphPlacementApproxBound:eq:potentialChange2}) follows from $\tilde{z}_{r}^{o}(i-1)\geq0$
and $\tilde{z}_{r}(i)=\tilde{z}_{r}(i-1)+\tilde{w}{}_{r|\pi_{i}}(i)$.
Note that the first term in (\ref{app:ProofGraphPlacementApproxBound:eq:potentialChange2}) is the same
as the objective function (\ref{app:ProofGraphPlacementApproxBound:eq:onlineExpObjWithElements}). Because
the mapping $\pi_{i}$ results from Algorithm \ref{algMultiLine}
which optimizes (\ref{app:ProofGraphPlacementApproxBound:eq:onlineExpObjWithElements}), we know that
the reference mapping $\pi_{0}$ must produce a cost $\alpha^{\tilde{z}_{r}(i-1)+\tilde{w}{}_{r|\pi_{i}^{o}}(i)}-\alpha^{\tilde{z}_{r}(i-1)}$
that is greater than or equal to the optimum, hence following (\ref{app:ProofGraphPlacementApproxBound:eq:potentialChange3}).
Equality (\ref{app:ProofGraphPlacementApproxBound:eq:potentialChange4}) is obvious.

Now we proof that the potential function $\Phi(i)$ does not increase
with $i$, by proving that (\ref{app:ProofGraphPlacementApproxBound:eq:potentialChange4}) is not larger than zero. For the $i$th request, the reference offline mapping produces
the mapping result $\pi_{i}^{o}$. Therefore, for all $r$ such that
$\exists\pi_{i}^{o}(\cdot)=r$, we have $0\leq\tilde{w}{}_{r|\pi_{i}^{o}}(i)\leq J_{\pi^{o}}/\hat{J}\leq1$.
Hence, we only need to show that $\gamma\left(\alpha^{\tilde{w}{}_{r|\pi_{i}^{o}}(i)}-1\right)-\tilde{w}{}_{r|\pi_{i}^{o}}(i)\leq 0$ for $\tilde{w}{}_{r|\pi_{i}^{o}}(i)\in[0,1]$,
which is true for $\alpha\leq1+1/\gamma$. From (\ref{app:ProofGraphPlacementApproxBound:eq:potentialChange1})--(\ref{app:ProofGraphPlacementApproxBound:eq:potentialChange4}), it follows that $\Phi(i)\leq\Phi(i-1)$. (We take $\alpha=1+1/\gamma$ because
this gives the smallest value of $\beta$.)

Because $\tilde{z}_{r}(0)=\tilde{z}_{r}^{o}(0)=0$, we have $\Phi(0)=\gamma(NK+L)$.
Because $\Phi(i)$ does not increase, $\alpha>1$, and $\tilde{z}_{r}^{o}(i)\leq1$
due to $J_{\pi^{o}}\leq\hat{J}$, we have 
\begin{eqnarray}
\left(\gamma-1\right)\alpha^{\max_{r}\tilde{z}_{r}(i)} & \leq & \left(\gamma-1\right)\sum_{r}\alpha^{\tilde{z}_{r}(i)} \nonumber \\
& \leq  & \Phi(i) \nonumber \\
& \leq & \Phi(0) \nonumber \\
&=&\gamma(NK+L).\label{app:ProofGraphPlacementApproxBound:eq:boundInequ}
\end{eqnarray}
Taking the logarithm on both sides of (\ref{app:ProofGraphPlacementApproxBound:eq:boundInequ}),
we have 
\begin{equation}
\max_{r}\tilde{z}_{r}(i)\leq\log_{\alpha}\left(\frac{\gamma(NK+L)}{\gamma-1}\right)=\beta,
\end{equation}
which proves the result because ${z}_{r}(i)=\tilde{z}_{r}(i)\cdot\hat{J}$.

\section*{Acknowledgment}

The authors gratefully thank Dr. Moez Draief, Dr. Ting He, Dr. Viswanath Nagarajan,  Dr. Theodoros Salonidis, and Dr. Ananthram Swami for their valuable suggestions to this work.

Contribution of S. Wang is partly related to his previous affiliation with Imperial College London. Contribution of M. Zafer is not related to his current employment with Nyansa Inc.

 \bibliographystyle{IEEEtran}
\bibliography{IEEEabrv,MappingAlgorithm}

\end{document}